\documentclass[a4paper,10pt,accepted=2021-11-11]{quantumarticle}
\pdfoutput=1

\usepackage[utf8]{inputenc}       % about time everything is utf-8 encoded
\usepackage[T1]{fontenc}          % encode special characters into font
\usepackage[english]{babel}       % we speak english
\usepackage{microtype}            % better typography
\usepackage{fnpct}                % kerning for footnotes
\usepackage{booktabs}             % nice tables
\usepackage[hidelinks]{hyperref}  % no ugly borders around links
\usepackage{silence}              % filter some useless warnings
\let\doi\undefined                % to use the doi package
\usepackage{doi}                  % to create doi links without "doi:"
\usepackage[numbers,sort&compress]{natbib}
\usepackage{caption}              % nicer figure captions

\usepackage{amsmath}
\usepackage{amssymb}
\usepackage{amsthm}
\usepackage{mathtools}
\usepackage{physics}              % for bra-ket and others
\usepackage{dsfont}               % for identity-1 and double-stroke letters
\usepackage{relsize}              % to make some symbols smaller
\usepackage{siunitx}              % for the S column in tabular
\usepackage[capitalise]{cleveref} % automated reference labelling
\let\vec\mathbf                   % bold vectors instead of arrows
\newtheorem{lemma}{Lemma}         % lemmas for appendix
\newtheorem{corollary}[lemma]{Corollary}  % corollaries in lemma-numbering

\usepackage{xcolor}
\usepackage{pgfplots}             % for plotting results
\usepackage{tikz}                 % for quantum circuit sketches
\pgfplotsset{compat=1.16}         % disable old compatibility stuff
\usetikzlibrary{decorations.pathreplacing,decorations.pathmorphing}  % for quantum circuits
\definecolor{col1}{RGB}{64,92,137} % plot color 1
\definecolor{col2}{RGB}{49,140,55} % plot color 2

\begin{document}
\title{Tailoring Term Truncations for Electronic Structure Calculations Using a Linear Combination of Unitaries}
\author{Richard Meister}
\affiliation{Department of Materials, University of Oxford, Oxford OX1 3PH, United Kingdom}
\author{Simon C. Benjamin}
\affiliation{Department of Materials, University of Oxford, Oxford OX1 3PH, United Kingdom}
\author{Earl T. Campbell}
\affiliation{Department of Physics and Astronomy, University of Sheffield, Sheffield S3 7RH, United Kingdom}
\affiliation{AWS Center for Quantum Computing, Pasadena, CA 91125, USA}

\begin{abstract}
\noindent A highly anticipated use of quantum computers is the simulation of complex quantum systems including molecules and other many-body systems. One promising method involves directly applying a linear combination of unitaries (LCU) to approximate a Taylor series by truncating after some order. Here we present an adaptation of that method, optimized for Hamiltonians with terms of widely varying magnitude, as is commonly the case in electronic structure calculations. We show that it is more efficient to apply LCU using a truncation that retains larger magnitude terms as determined by an iterative procedure. We obtain bounds on the simulation error for this generalized truncated Taylor method, and for a range of molecular simulations, we report these bounds as well as exact numerical results. We find that our adaptive method can typically improve the simulation accuracy by an order of magnitude, for a given circuit depth.
\end{abstract}

\maketitle

\section{Introduction}
One of the most promising applications of quantum computers is the efficient simulation of quantum systems~\cite{feynman1982simulating}, including those that arise in quantum chemistry. Following the first concepts for such simulations~\cite{lloyd1996universal,abrams1997sim}, there have been numerous proposed algorithms to simulate these systems using quantum computers~\cite{ortiz2001quantum,aspuru2005simulated,berry2006efficient,wang2008algorithm,whitfield2010simulation,campbell2019random,childs2019nearly,childs2019faster,ouyang2020compilation}, often with variations of Trotter-Suzuki product formulas~\cite{trotter1959product,suzuki1976product}. These methods usually approximate the time evolution operator by sequentially evolving the terms in the Hamiltonian individually. Through extensive study, the required gate count was reduced substantially over time~\cite{wecker2014gatecount,babbush2015chemical,hastings2014improving,poulin2014trotter,mcclean2014locality,childs2019trotter}. However, the scaling of the inverse simulation error of such product formulas is polynomial in the circuit gate count.

An alternative is available through the technique of linear combinations of unitaries (LCU). Here, in contrast to the product formula approaches, one derives a quantum circuit that directly applies a sum of unitaries, allowing for a much greater variety of accessible operators. A key enhancement was the replacement of a probabilistic step in the original scheme~\cite{childs2012lcu} with a near-deterministic process based on oblivious amplitude amplification~\cite{berry2014improvement}.

The LCU method gave rise to a number of implementations for Hamiltonian simulation. The approach in Ref.~\cite{low2019wellconditioned} uses linear combinations of product formulas, taking advantage of commuting terms in the Hamiltonian -- like pure product formulas -- while improving the complexity scaling with inverse error to be only poly-logarithmic using LCU. In Ref.~\cite{berry2015hamiltonian} it is applied to enhance the scaling with the error in Szegedy quantum walks while retaining their advantage for sparse Hamiltonians. Extensions of this approach are quantum signal processing~\cite{low2017optimal,gilyen2019quantum} and qubitization~\cite{low2019qubitization}, of which variants specifically for quantum chemistry exist~\cite{berry2019qubitization}.

One of the most direct uses of LCU is presented in~\cite{berry2015taylor}, where the time evolution operator is approximated by truncating its Taylor expansion at some appropriate order. This results in exponentially better scaling of the complexity with inverse error than for product formulas.

The aforementioned methods of qubitization and quantum signal processing have been shown to exhibit even better scaling for many types of Hamiltonians~\cite{childs2018speedup,babbush2018encoding,low2019qubitization,berry2019qubitization}. However, there are instances where they are less suited, one prominent example being for simulating time-dependent Hamiltonians. Even for intrinsically time-independent cases, introducing a time dependence by transforming to a rotating frame can be beneficial if the Hamiltonian is diagonally dominant. In contrast to qubitization and quantum signal processing, the approach of the truncated Taylor series in~\cite{berry2015taylor} can be applied to such time-dependent cases with reasonable overhead, as shown in Refs.~\cite{low2018hamiltonian,babbush2019chemistry}, making it very relevant for such instances.

In this work, we present a variant of the truncated Taylor scheme~\cite{berry2015taylor} that includes terms by weight rather than order. By doing so, we try to exploit the fact that the Hamiltonians of some quantum mechanical systems~\nobreakdash--~especially those of electrons in molecules~\nobreakdash--~have terms whose magnitudes vary considerably. This suggests that some (large) terms should be included to higher orders than other (small) terms. Our variant implements just that, while respecting the efficient circuit implementation of~\cite{berry2015taylor} and subsequent improvements to \textsc{select} and \textsc{prepare} subroutines~\cite{babbush2018encoding,low2018trading}.

Our algorithm starts from an empty expansion and iteratively adds terms that facilitate the largest decline of the error bound for one additional gate. This greedy method leads to a more rapid reduction of the error in the very early stage of the construction when applied to electronic structure Hamiltonians. At a later stage, the rate of convergence becomes roughly equal to the original method, maintaining an approximately constant factor advantage in the error for the investigated molecules. Therefore, the asymptotic behavior is equivalent for both methods, but we accomplish a constant improvement. We find that the error of our modified scheme is typically one order of magnitude lower than in the original method at the same gate cost. For a fixed error magnitude, this results in reducing the circuit depth by roughly one full order of the expansion.

The rest of the paper is structured as follows. \Cref{sec:trunc_taylor} contains a detailed description of our modified method adapted from~\cite{berry2015taylor}. In \cref{sec:results}, we present results for error bounds as well as numerically evaluated errors for a variety of electronic structure Hamiltonians. Lastly, \cref{sec:conclusion} concludes the paper and gives an outlook to possible further work.

\section{Truncated Taylor series} \label{sec:trunc_taylor}
Our method is closely related to the approach presented by~\citet{berry2015taylor}. We will give a detailed description of our modified method, which at the same time serves as a summary of~\cite{berry2015taylor}.
\subsection{Linear combination of unitaries}
The protocol is based on a method of adding unitaries with the help of ancilla qubits~\cite{childs2012lcu}. We start from a Hamiltonian of the form
\begin{equation} \label{eq:H}
    H = \sum_{\ell=0}^{L-1}\alpha_\ell h_\ell,
\end{equation}
where $\alpha_\ell$ are real positive scalars\footnote{Phases can always be pushed into the operators $h_\ell$.} and $h_\ell$ are unitaries for which implementations on a quantum computer exist. Without loss of generality, we assume the terms are sorted by magnitude, i.e. $\alpha_{\ell + 1} \geq \alpha_\ell$. The approach also used in~\cite{berry2015taylor} is to implement an approximation to the corresponding time evolution operator
\begin{equation}
    U(t) = e^{-iHt}
\end{equation}
with a Taylor series. Taking $t$ to be sufficiently small, the series representation of $U(t)$ can be approximated by the sum
\begin{equation} \label{eq:U_L}
    U_{\vec{L}}(t) \coloneqq \mathds{1} + \sum_{k=1}^{\infty} \frac{(-it)^k}{k!}\, \prod_{j = 1}^{k}\left(\sum_{\ell_j = 0}^{L_j - 1} \alpha_{\ell_j} h_{\ell_j}\right)
\end{equation}
where $\vec{L}$ is a vector of $L_k$ with $0 \leq L_k \leq L$ and $k\in\mathds{N}^+$, meaning the individual sums in the product only contain the $L_k$ largest terms of $H$. This is the main difference to~\cite{berry2015taylor}, where the series is truncated at some appropriate order $n$, which yields
\begin{equation} \label{eq:U_n}
    U_n(t) \coloneqq \mathds{1} + \sum_{k=1}^{n} \frac{(-it)^k}{k!}\, \prod_{j = 1}^{k}\left(\sum_{\ell_j = 0}^{L - 1} \alpha_{\ell_j} h_{\ell_j}\right).
\end{equation}
\Cref{eq:U_n} is a special case of \cref{eq:U_L}, where all orders up to $n$ are added in full\footnote{For all quantities with an $\vec{L}$ subscript we will alternatively replace it with $n$ to mean an $\vec{L}$ where $L_k = L$ for $k \leq n$ and $L_k = 0$ for $k > n$.}. Our modified version of the sum includes some orders only partially, giving greater control over the total gate count and allowing for quicker convergence of the error bounds.

The magnitude of the time step $t$ will be a fixed value restricted by the method. Longer times $\tau = rt$, can be simulated by applying $U_{\vec{L}}^r$. However, most of this paper will focus on the implementation of a single time step.

To keep the notation simple, the products of the coefficients $\alpha_\ell$ with $t^k/k!$ are gathered into new variables $\beta_j$, and all products of the unitaries $h_\ell$ together with $(-i)^k$ are collected into operators $V_j$, with a newly introduced label $j$ numbering all terms in the sum. Note that even if different products of $h_\ell$ yield identical operators, they are treated as separate $V_j$, each with a corresponding weight $\beta_j$.
By construction, all $\beta_j$ are also real and positive. Thus, \cref{eq:U_L} becomes
\begin{equation}
    U_{\vec{L}} = \sum_{j=0}^{m-1} \beta_j V_j
\end{equation}
where the time-dependence of $U_{\vec{L}}$ and $\beta_j$ is not explicitly denoted, and the total number of terms $m$ implicitly depends on $\vec{L}$.

In order to apply $U_{\vec{L}}$ to a state $\ket{\psi}$, we define the unitary operators $\mathcal{P}(t)$ and $\mathcal{S}$ (\textsc{prepare} and \textsc{select}) in accordance with~\cite{berry2015taylor}. The \textsc{prepare} operator $\mathcal{P}$, whose time dependence we will make implicit from here on, maps the $\ket{0}$ state of a set of ancilla qubits to the weighted superposition
\begin{equation} \label{eq:prepare}
    \mathcal{P}  \ket{0} \coloneqq \frac{1}{\sqrt{s_{\vec{L}}}} \sum_{j=0}^{m-1}\sqrt{\beta_j}\ket{j}
\end{equation}
with the implicitly $t$-dependent normalization constant
\begin{equation}
    s_{\vec{L}} \coloneqq \sum_{j=0}^{m-1} \beta_j.
\end{equation}

The \textsc{select} operator $\mathcal{S}$ acts on a state $\ket{\psi}$ with the operator $V_j$, where $j$ is the state of the ancilla introduced above. So its action on a tensor state of $\ket{\psi}$ with the ancilla $\ket{j}$ is
\begin{equation}
    \mathcal{S} \ket{j}\ket{\psi} \coloneqq \ket{j} V_j \ket{\psi}.
\end{equation}
Given these two operators $\mathcal{P}$ and $\mathcal{S}$, we proceed analogously to~\cite{berry2015taylor} by introducing a new operator
\begin{equation}
    \mathcal{W} \coloneqq (\mathcal{P}^\dagger \otimes \mathds{1}) \, \mathcal{S} \, (\mathcal{P} \otimes \mathds{1})
\end{equation}
which has the effect
\begin{equation}\label{eq:W}
    \mathcal{W} \ket{0}\ket{\psi} = \frac{1}{s_{\vec{L}} }\ket{0} U_{\vec{L}} \ket{\psi} + N \ket{0^\mathsmaller{\perp}, \Phi}
\end{equation}
where $N$ is the appropriate constant for the state to be normalized, and $\ket{0^\mathsmaller{\perp}, \Phi}$ is a garbage state whose ancilla part has no overlap with the ancillary $\ket{0}$ state.

\subsection{Oblivious amplitude amplification}
The naïve method for obtaining $U_{\vec{L}}\ket{\psi}$ would be to measure the ancilla of $\mathcal{W}\ket{0}\ket{\psi}$, see \cref{eq:W}, and post-select for the ancilla $\ket{0}$ state. However, since $s_{\vec{L}}$ increases with $t$, the success probability for large $t$ diminishes. Additionally, $t$ is always subject to convergence of \cref{eq:U_L}. Due to the postselection, dividing the total $t$ into smaller segments and repeating the process multiple times would also suppress the total success probability.

One way around this problem also used in~\cite{berry2015taylor} is the so-called oblivious amplitude amplification. As detailed in \cref{lemma:amp} in \cref{sec:derivations}, and references therein, if $U_{\vec{L}}$ were unitary and $s_\vec{L} = 2$, the amplification operator
\begin{equation} \label{eq:q}
    \mathcal{Q} \coloneqq -\mathcal{W} R \mathcal{W}^\dagger R
\end{equation}
with $R\coloneqq 2 \Pi - \mathds{1}$ the reflection operator about the $\ket{0}$ state of the ancilla and $\Pi \coloneqq \ketbra{0}{0} \otimes \mathds{1}$ the projector onto the ancilla $\ket{0}$, would have the effect~\cite{berry2014improvement}
\begin{align} \label{eq:amp}
    \mathcal{Q}\mathcal{W} \ket{0} \ket{\psi}= \ket{0} U_{\vec{L}} \ket{\psi}.
\end{align}
Thus, we define
\begin{equation}
    \mathcal{A} \coloneqq \mathcal{Q}\mathcal{W} = -\mathcal{W}R\mathcal{W}^\dagger R\mathcal{W}.
\end{equation}

We first discuss the requirement of $s_{\vec{L}} = 2$. Our form of the Taylor expansion leads to $s_{\vec{L}}$ being of the form
\begin{equation} \label{eq:s_L}
    s_{\vec{L}}(t) \coloneqq \sum_{k=1}^\infty \frac{t^k}{k!}\,\prod_{j=1}^{k}\,\smash{\underbrace{\!\left[\sum_{\ell_j=0}^{L_j-1} \alpha_{\ell_j}\right]\!}_{\coloneqq \Lambda_j}}= \sum_{k=0}^\infty \frac{t^k}{k!}\prod_{j=1}^k \Lambda_j.\vphantom{\underbrace{\sum_j^L}_{A_j}}
\end{equation}
The restriction $s_{\vec{L}} = 2$ therefore forces the simulation time $t$ to be the only real root of
\begin{equation}
    \sum_{k=0}^\infty \frac{t^k}{k!}\prod_{j=1}^k \Lambda_j - 2 = 0
\end{equation}
which we call $t_{\vec{L}}$. If we were to include all orders in full, i.e. $L_k = L, \ \forall\,k$, all $\Lambda_j$ would be equal and the infinite sum on the left would become the series of the exponential function. We call the time step for this case $t_{\infty} = \log{(2)}/\Lambda$, with the definition $\Lambda \coloneqq \sum_{j=0}^{L} \alpha_j$.

Shorter times can be accomplished by using an extra qubit, as described in~\cite{berry2014improvement}. Since the only requirement for oblivious amplitude amplification to work is $s_{\vec{L}} = 2$, and shorter times mean $s_{\vec{L}} < 2$, i.e. the amplitude of the ancilla $\ket{0}$ is too large, we can introduce an additional qubit to the ancilla and prepare it with enough weight such that the ancilla $\ket{0}$ reduces to amplitude $1/2$. These shorter times are only relevant in the last time step of a simulation and have almost the same cost as a full step, so we limit the rest of the discussion to multiples of $t_{\vec{L}}$.

\Cref{eq:amp} only strictly holds for unitary $U_{\vec{L}}$, but the series truncation means that $U_{\vec{L}}$ is only approximately unitary.  Therefore, we need the action of $\mathcal{A}$ for a general $U_{\vec{L}}$ and again follow~\cite{berry2015taylor}. Applying $\mathcal{A}$ to a state $\ket{0}\ket{\psi}$ and projecting onto the ancilla $\ket{0}$ yields
\begin{equation}
    \Pi\mathcal{A}\ket{0}\ket{\psi} = \ket{0} \left(\frac{3}{s_{\vec{L}}}U_{\vec{L}} - \frac{4}{s_{\vec{L}}^3}U_{\vec{L}}^{\vphantom{\dagger}} U_{\vec{L}}^\dagger U_{\vec{L}}^{\vphantom{\dagger}}\right) \ket{\psi},
\end{equation}
(derivation in \cref{sec:derivations}, \cref{lemma:action_A}) and we call the operator we are actually applying in the $\ket{\psi}$ subspace
\begin{equation} \label{eq:a_tilde}
    \tilde{\mathcal{A}}_{\vec{L}} \coloneqq \frac{3}{s_{\vec{L}}}U_{\vec{L}} - \frac{4}{s_{\vec{L}}^3}U_{\vec{L}}^{\vphantom{\dagger}} U_{\vec{L}}^\dagger U_{\vec{L}}^{\vphantom{\dagger}}.
\end{equation}

\subsection{Gate construction}
We also want to elaborate on the specific gate construction to implement $\mathcal{A}$ efficiently, adapted from~\cite{berry2015taylor}. First, the ancilla is divided into $\kappa + 1$ registers, where $\kappa \coloneqq ||\vec{L}||_0$ is the number of non-zero elements in the vector $\vec{L}$. The first register is named $q$ and contains $\kappa$ qubits, while the others are given labels $c_1 \ldots c_\kappa$, with $c_k$ containing $\lceil \log_2 L_k \rceil$ qubits.

The $q$ register's purpose is to represent different orders, while registers $c_k$ are needed for the terms in each order. This makes it convenient to use a multi-index $j \equiv (k, \ell_1, \ldots, \ell_k)$. The corresponding state of the ancilla is
\begin{equation}
    \ket{j} \equiv \ket{k}_{\!q}\ket{\ell_1}_{\!c_1}\ldots\ket{\ell_k}_{\!c_k}\ldots
\end{equation}
where we leave the state of the registers $c_{k'}$ with $k' > k$ unspecified. The coefficient associated with this index is
\begin{equation}
    \beta_j = \beta_{(k, \ell_1, \ldots, \ell_k)} = \frac{t^k}{k!}\,\alpha_{\ell_1}\ldots\alpha_{\ell_k}.
\end{equation}

\paragraph{\textsc{prepare}}
For this operator, we slightly deviate from~\cite{berry2015taylor}. Exact implementation of $\mathcal{P}$ as defined in \cref{eq:prepare} would necessitate the preparation of the $c_k$ registers to be conditioned on qubits in the $q$ register. We can, however, implement an operator $\mathcal{P}^\star$, which acts equivalently to $\mathcal{P}$ when used in $\mathcal{W}$, but is performed independently on each of the $c_k$ registers without controls on the qubits in $q$.

The $q$ register will contain the prefactor for each order $k$ and uses unary coding, i.e. $\smash{\ket{k}_{\!q} \coloneqq \ket{1^k 0^{k-\kappa}}_{\!q}}$. Thus, the \textsc{prepare} operator $\mathcal{P}^\star(t)$ acts on this register proportional to

\begin{equation}
    \ket{0^\kappa}_{\!q} \mapsto \sum_{k=0}^\kappa \sqrt{\frac{t^k}{k!}\prod_{j=1}^k \Lambda_j} \ket{k}_{\!q}.
\end{equation}
This can be implemented by a rotation on the first qubit, and rotations controlled by the previous one on each subsequent qubit.

The $c_k$ registers can now all be almost identically prepared to contain the coefficients of the Hamiltonian, where each index $\ell$ is mapped to the qubits of $c_k$ in regular binary coding. So the action of $\mathcal{P}^\star$ on a single register $c_k$ is proportional to
\begin{equation}
    \ket{0}_{\!c_k} \mapsto \sum_{\ell=0}^{L_k-1} \sqrt{\alpha_\ell} \ket{\ell}_{\!c_k}.
\end{equation}
For this, any efficient method for arbitrary state preparation can be used, whose cost we discuss presently.

Combining these constituents into a single unitary $\mathcal{P}^\star$ and applying it to the whole ancilla yields the desired operator equivalent to \cref{eq:prepare} if used in $\mathcal{W}$, which is shown in more detail in \cref{lemma:p_star} in \cref{sec:derivations}.

\paragraph{\scshape select}
Using the established structure of the ancilla, the $\mathcal{S}$ operator must have the action
\begin{multline}
    \mathcal{S} \ket{k}_{\!q} \ket{\ell_1}_{\!c_1}\ldots \ket{\ell_k}_{\!c_k} \ldots \ket{\ell_\kappa}_{\!c_\kappa} \ket{\psi} \\ = \ket{k}_{\!q} \ket{\ell_1}_{\!c_1}\ldots \ket{\ell_k}_{\!c_k} \ldots \ket{\ell_\kappa}_{\!c_\kappa} \tilde{h}_{\ell_1} \ldots \tilde{h}_{\ell_k} \ket{\psi}
\end{multline}
with $\tilde{h}_{\ell} \coloneqq -ih_\ell$. This can be accomplished by having a sequence of groups of unitaries in the circuit\footnote{The groups in the circuit are numbered right-to-left to match the established numbering convention of the operators.}. Each of the groups $m = 1\ldots \kappa$ contains the unitaries $\tilde{h}_{\ell_m}$, with $\ell_m = 0\ldots L_m - 1$, acting on the target state $\ket{\psi}$.

The register $c_m$ is used as the addressing register for group $m$, i.e. the state $\ket{\ell_m}_{\!c_m}$ determines which unitary in group $m$ is applied. To achieve this, we use the fact that the $c$ registers are in binary coding, so $\ell_m$ is represented as a binary number with the $\lceil\log_2 L_m \rceil$ qubits in $c_m$ as digits. By controlling $\smash{\tilde{h}_{\ell_m}}$ on the $c_m$ register in a way that matches the binary representation of $\ell_m$, only the unitary with the correct index is applied. For example, $\tilde{h}_5$ would be controlled by the last and antepenultimate qubit in $c_m$ and anti-controlled by all other qubits in $c_m$ (since 5 corresponds to the state $\ket{0\ldots 0101}$ in binary coding).

Additionally, the $q$ register specifies how many of the groups are applied. If $q$ is in the state $\ket{k}_{\!q}$, only the first $k$ groups should be active. The unary coding in $q$ makes this straightforward to implement by additionally controlling every unitary in group $m$ with the $m\textsuperscript{th}$ qubit in $q$. \Cref{fig:select} shows a sketch of the full construction.

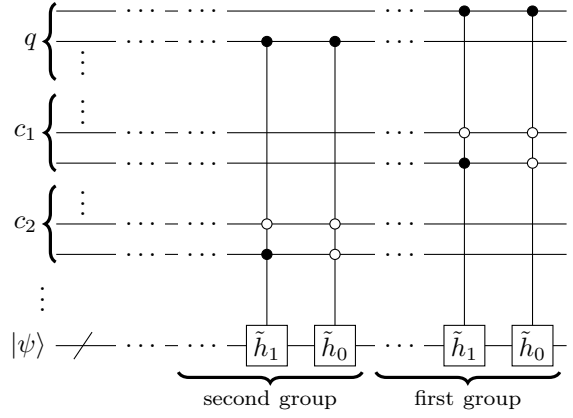
\begin{figure}[tbh]
    \centering
    %! \usepackage{mathtools}
%! \usetikzlibrary{decorations.pathreplacing,decorations.pathmorphing}
\providecommand{\ket}[1]{\left|#1\right\rangle}
\begin{tikzpicture}[scale=0.850000,x=1pt,y=1pt]
\draw[use as bounding box,draw=none] (-20.000000, -30.0000) rectangle (225.000000, 155.250000);
% Drawing wires
% Line 5: k1 k2 W
\draw[color=black] (0.000000,148.500000) -- (225.000000,148.500000);
%   Deferring wire label at (0.000000,148.500000)
% Line 5: k1 k2 W
\draw[color=black] (0.000000,135.000000) -- (225.000000,135.000000);
%   Deferring wire label at (0.000000,135.000000)
% Line 6: k3 W owire
\filldraw[color=white,fill=white] (0.000000,118.125000) rectangle (-4.000000,151.875000);
\draw[decorate,decoration={brace,amplitude = 4.000000pt},very thick] (0.000000,118.125000) -- (0.000000,151.875000);
\draw[color=black] (-4.000000,135.000000) node[left] {$q$};
% Line 8: l13 W owire
%   Deferring wire label at (0.000000,108.000000)
% Line 9: l11 l12 W
\draw[color=black] (0.000000,94.500000) -- (225.000000,94.500000);
%   Deferring wire label at (0.000000,94.500000)
% Line 9: l11 l12 W
\draw[color=black] (0.000000,81.000000) -- (225.000000,81.000000);
\filldraw[color=white,fill=white] (0.000000,77.625000) rectangle (-4.000000,111.375000);
\draw[decorate,decoration={brace,amplitude = 4.000000pt},very thick] (0.000000,77.625000) -- (0.000000,111.375000);
\draw[color=black] (-4.000000,94.500000) node[left] {$c_1$};
% Line 11: l23 W owire
%   Deferring wire label at (0.000000,67.500000)
% Line 12: l21 l22 W
\draw[color=black] (0.000000,54.000000) -- (225.000000,54.000000);
%   Deferring wire label at (0.000000,54.000000)
% Line 12: l21 l22 W
\draw[color=black] (0.000000,40.500000) -- (225.000000,40.500000);
\filldraw[color=white,fill=white] (0.000000,37.125000) rectangle (-4.000000,70.875000);
\draw[decorate,decoration={brace,amplitude = 4.000000pt},very thick] (0.000000,37.125000) -- (0.000000,70.875000);
\draw[color=black] (-4.000000,54.000000) node[left] {$c_2$};
% Line 14: l31 W owire \substack{\vspace{.4em}\\ \vdots}
\draw[color=black] (0.000000,27.000000) node[left] {$\substack{\vspace{.4em}\\ \vdots}$};
% Line 15: l41 W owire
% Line 16: psi W \ket{\psi}
\draw[color=black] (0.000000,0.000000) -- (225.000000,0.000000);
\draw[color=black] (0.000000,0.000000) node[left] {$\ket{\psi}$};
% Done with wires; drawing gates
% Line 18: k3 G $\substack{\vdots\\ \vspace{.7em}}$ shape=0
\begin{scope}
\draw (12.000000, 121.500000) node {$\substack{\vdots\\ \vspace{.7em}}$};
\end{scope}
% Line 19: l13 G $\substack{\vdots\\ \vspace{-.7em}}$ shape=0
\begin{scope}
\draw (12.000000, 108.000000) node {$\substack{\vdots\\ \vspace{-.7em}}$};
\end{scope}
% Line 20: l23 G $\substack{\vdots\\ \vspace{-.7em}}$ shape=0
\begin{scope}
\draw (12.000000, 67.500000) node {$\substack{\vdots\\ \vspace{-.7em}}$};
\end{scope}
% Line 21: psi /
\draw (8.000000, -6.000000) -- (16.000000, 6.000000);
% Line 23: LABEL \hdots
\draw[color=black] (37.500000, 148.500000) node [fill=white] {$\hdots$};
\draw[color=black] (37.500000, 135.000000) node [fill=white] {$\hdots$};
\draw[color=black] (37.500000, 94.500000) node [fill=white] {$\hdots$};
\draw[color=black] (37.500000, 81.000000) node [fill=white] {$\hdots$};
\draw[color=black] (37.500000, 54.000000) node [fill=white] {$\hdots$};
\draw[color=black] (37.500000, 40.500000) node [fill=white] {$\hdots$};
\draw[color=black] (37.500000, 0.000000) node [fill=white] {$\hdots$};
% Line 25: LABEL \hdots
\draw[color=black] (64.500000, 148.500000) node [fill=white] {$\hdots$};
\draw[color=black] (64.500000, 135.000000) node [fill=white] {$\hdots$};
\draw[color=black] (64.500000, 94.500000) node [fill=white] {$\hdots$};
\draw[color=black] (64.500000, 81.000000) node [fill=white] {$\hdots$};
\draw[color=black] (64.500000, 54.000000) node [fill=white] {$\hdots$};
\draw[color=black] (64.500000, 40.500000) node [fill=white] {$\hdots$};
\draw[color=black] (64.500000, 0.000000) node [fill=white] {$\hdots$};
% Line 27: psi G $\tilde{h}_1$ size=18 k2 -l21  l22
\draw (93.000000,135.000000) -- (93.000000,0.000000);
\begin{scope}
\draw[fill=white] (93.000000, -0.000000) +(-45.000000:12.727922pt and 12.727922pt) -- +(45.000000:12.727922pt and 12.727922pt) -- +(135.000000:12.727922pt and 12.727922pt) -- +(225.000000:12.727922pt and 12.727922pt) -- cycle;
\clip (93.000000, -0.000000) +(-45.000000:12.727922pt and 12.727922pt) -- +(45.000000:12.727922pt and 12.727922pt) -- +(135.000000:12.727922pt and 12.727922pt) -- +(225.000000:12.727922pt and 12.727922pt) -- cycle;
\draw (93.000000, -0.000000) node {$\tilde{h}_1$};
\end{scope}
\filldraw (93.000000, 135.000000) circle(2.250000pt);
\draw[fill=white] (93.000000, 54.000000) circle(2.250000pt);
\filldraw (93.000000, 40.500000) circle(2.250000pt);
% Line 28: psi G $\tilde{h}_0$ size=18 k2 -l21 -l22
\draw (123.000000,135.000000) -- (123.000000,0.000000);
\begin{scope}
\draw[fill=white] (123.000000, -0.000000) +(-45.000000:12.727922pt and 12.727922pt) -- +(45.000000:12.727922pt and 12.727922pt) -- +(135.000000:12.727922pt and 12.727922pt) -- +(225.000000:12.727922pt and 12.727922pt) -- cycle;
\clip (123.000000, -0.000000) +(-45.000000:12.727922pt and 12.727922pt) -- +(45.000000:12.727922pt and 12.727922pt) -- +(135.000000:12.727922pt and 12.727922pt) -- +(225.000000:12.727922pt and 12.727922pt) -- cycle;
\draw (123.000000, -0.000000) node {$\tilde{h}_0$};
\end{scope}
\filldraw (123.000000, 135.000000) circle(2.250000pt);
\draw[fill=white] (123.000000, 54.000000) circle(2.250000pt);
\draw[fill=white] (123.000000, 40.500000) circle(2.250000pt);
% Line 30: LABEL \hdots
\draw[color=black] (151.500000, 148.500000) node [fill=white] {$\hdots$};
\draw[color=black] (151.500000, 135.000000) node [fill=white] {$\hdots$};
\draw[color=black] (151.500000, 94.500000) node [fill=white] {$\hdots$};
\draw[color=black] (151.500000, 81.000000) node [fill=white] {$\hdots$};
\draw[color=black] (151.500000, 54.000000) node [fill=white] {$\hdots$};
\draw[color=black] (151.500000, 40.500000) node [fill=white] {$\hdots$};
\draw[color=black] (151.500000, 0.000000) node [fill=white] {$\hdots$};
% Line 32: psi G $\tilde{h}_1$ size=18 k1 -l11  l12
\draw (180.000000,148.500000) -- (180.000000,0.000000);
\begin{scope}
\draw[fill=white] (180.000000, -0.000000) +(-45.000000:12.727922pt and 12.727922pt) -- +(45.000000:12.727922pt and 12.727922pt) -- +(135.000000:12.727922pt and 12.727922pt) -- +(225.000000:12.727922pt and 12.727922pt) -- cycle;
\clip (180.000000, -0.000000) +(-45.000000:12.727922pt and 12.727922pt) -- +(45.000000:12.727922pt and 12.727922pt) -- +(135.000000:12.727922pt and 12.727922pt) -- +(225.000000:12.727922pt and 12.727922pt) -- cycle;
\draw (180.000000, -0.000000) node {$\tilde{h}_1$};
\end{scope}
\filldraw (180.000000, 148.500000) circle(2.250000pt);
\draw[fill=white] (180.000000, 94.500000) circle(2.250000pt);
\filldraw (180.000000, 81.000000) circle(2.250000pt);
% Line 33: psi G $\tilde{h}_0$ size=18 k1 -l11 -l12
\draw (210.000000,148.500000) -- (210.000000,0.000000);
\begin{scope}
\draw[fill=white] (210.000000, -0.000000) +(-45.000000:12.727922pt and 12.727922pt) -- +(45.000000:12.727922pt and 12.727922pt) -- +(135.000000:12.727922pt and 12.727922pt) -- +(225.000000:12.727922pt and 12.727922pt) -- cycle;
\clip (210.000000, -0.000000) +(-45.000000:12.727922pt and 12.727922pt) -- +(45.000000:12.727922pt and 12.727922pt) -- +(135.000000:12.727922pt and 12.727922pt) -- +(225.000000:12.727922pt and 12.727922pt) -- cycle;
\draw (210.000000, -0.000000) node {$\tilde{h}_0$};
\end{scope}
\filldraw (210.000000, 148.500000) circle(2.250000pt);
\draw[fill=white] (210.000000, 94.500000) circle(2.250000pt);
\draw[fill=white] (210.000000, 81.000000) circle(2.250000pt);
% Done with gates; drawing ending labels
% Done with ending labels; drawing cut lines and comments
% Line 36: @ 2 4 %% \footnotesize second group
\draw[decorate,decoration={brace,mirror,amplitude = 4.000000pt},very thick] (54.000000,-11.750000) -- (135.000000,-11.750000);
\draw (94.500000, -15.750000) node[text width=144pt,below,text centered] {\footnotesize second group};
% Line 37: @ 5 7 %% \footnotesize first group
\draw[decorate,decoration={brace,mirror,amplitude = 4.000000pt},very thick] (141.000000,-11.750000) -- (222.000000,-11.750000);
\draw (181.500000, -15.750000) node[text width=144pt,below,text centered] {\footnotesize first group};
% Done with comments
\end{tikzpicture}
    \vspace{2ex}
    \caption{Sketch of the gate construction for $\mathcal{S}$. By taking advantage of the unary iteration structure, the $T$-count of the multi-controls can be significantly reduced~\cite{babbush2018encoding}. However, we include this non-optimized diagram for pedagogical purposes.}
    \label{fig:select}
    \vspace{-1ex}
\end{figure}

\paragraph{Gate cost}
Lastly, we want to estimate the gate complexity of the operator $\mathcal{A}$. Its constituents are two reflections $R$, and three instances of $\mathcal{W}$, each of which contains one $\mathcal{S}$ and two $\mathcal{P}^\star$. For calculations using full orders as in~\cite{berry2015taylor}, our analysis translates exactly to the gate construction given there.  We consider the universal set of Clifford + $T$ and count the number of expensive $T$-gates~\cite{bravyi2005magic,howard2014magic,campbell2017magic} for our complexity analysis.

Each reflection $R$ is a single Pauli-$Z$ operator on one of the ancilla qubits (padded between two \textsc{not} gates), anti-controlled on all others. This can be done with $\mathcal{O}(\sum_k \log_2 L_k)$ $T$-gates and a second ancilla register of size $(\kappa + \sum_k \lceil \log_2 L_k \rceil - 2)$~\cite{nielsen_chuang}.

The \textsc{prepare} stage for the $q$ register consists of $\kappa - 1$ controlled rotations with a total $T$-complexity of $\mathcal{O}(\kappa)$. Each of the $\kappa$ registers $c_k$ needs to be initialized to a specific state with $2^{\lceil \log_2 L_k \rceil} \sim L_k$ coefficients, requiring between $\mathcal{O}(\sum_k L_k)$ and $\mathcal{O}(\sum_k \sqrt{L_k}\log^2(L_k/\epsilon))$ $T$\nobreakdash-gates per register, depending on the number of additionally available ancillas, where $\epsilon$ is the accuracy of the preparation~\cite{low2018trading}. In total, this yields a $T$-count between $\mathcal{O}(\sum_k L_k)$ and $\mathcal{O}(\sum_k \sqrt{L_k} \log^2(L_k/\epsilon))$.

The fact that the controls of each $h_\ell$ in $\mathcal{S}$ form a so-called unary iteration can be exploited to lower the $T$-gate count. Each sequence of $L_k$ operators can be implemented using $\mathcal{O}(L_k)$ $T$-gates~\cite{babbush2018encoding}, plus $L_k$ times the cost of performing a single $-ih_\ell$ operator, totalling to $\sum_k L_k$ such operators. Thus, the $T$-complexity of $\mathcal{S}$ for generic Hamiltonians in the form of \cref{eq:H} is of $\mathcal{O}(\sum_k L_k)$, which we will use in this paper. Moreover, recent work~\cite{berry2019qubitization} has shown that a \textsc{select} process can be yet more efficient for the special case of $N$-orbital electronic structure problems, where the $T$-complexity is as low as $\mathcal{O}(nN)$, with $n = \max\{k : L_k \neq 0\}$.

Combining all these counts results in a total complexity of $\mathcal{O}(\sum_k L_k)$ for $\mathcal{A}$. As a proxy to use for the total gate cost in our results we thus define
\begin{equation} \label{eq:c}
    C_{\vec{L}} \coloneqq \sum_{k=1}^\infty L_k = ||\vec{L}||_1.
\end{equation}
This definition includes the cost of a full expansion to order $n$ as the special case $C_n = nL$, consistent with our previous notation. From this cost of a single time step, we discuss the complexity $C_\epsilon$ to reach some desired total simulation error $\epsilon$ in the next subsection.

\subsection{Error bounds}
We consider the error of the method per time step to be the norm of an operator $\Delta_{\vec{L}}$ which fulfills
\begin{equation}
    U(t_\infty) = \tilde{\mathcal{A}}_{\vec{L}}(t_\infty) + \Delta_{\vec{L}}(t_\infty)
\end{equation}
where we now use the step size $t_\infty = \log(2)/\Lambda$. We find that the error made by applying $\Pi\mathcal{A}$ once and tracing out the ancilla can be bounded by\footnote{$||\cdot||$ in this paper always means the operator norm.}
\begin{equation}
    \delta_{\vec{L}} \coloneqq ||\Delta_{\vec{L}}(t_\infty)|| \leq 2 - s_{\vec{L}}(t_\infty) \eqqcolon \varepsilon_{\vec{L}}
\end{equation}
up to order $\varepsilon_{\vec{L}}$ (details in \cref{sec:derivations}, \cref{lemma:delta_L}). Because using $t_{\vec{L}}$ or $t_\infty$ makes no difference in the error up to order $\varepsilon_{\vec{L}}$, we exclusively use $t_\infty$ in our calculations. The error for a total simulation time $\tau = rt_\infty = r \log(2) / \Lambda$, $r \in \mathds{N}$, is then
\begin{equation}
    ||\tilde{\mathcal{A}}_{\vec{L}}(t_\infty)^r - U(t_\infty)^r || = r\delta_{\vec{L}} = \frac{\Lambda \delta_{\vec{L}}}{\log 2}\,\tau \leq r\varepsilon_{\vec{L}},
\end{equation}
also up to order $\varepsilon_{\vec{L}}$ (see \cref{sec:derivations}, \cref{lemma:error_repetitions}).

We call the bound on the total simulation error of $r$ steps $\epsilon \coloneqq r\varepsilon$. The $T$-gate complexity $C_\epsilon$ of a simulation for time $\tau$ in terms of the total error bound $\epsilon$ is then in the range
\begin{equation}
    \mathcal{O}\left(\frac{\Lambda \tau \log\frac{\Lambda\tau}{\epsilon}}{\log\log\frac{\Lambda\tau}{\epsilon}}\right) < C_\epsilon \leq \mathcal{O}\left(\frac{L\Lambda\tau\log\frac{\Lambda\tau}{\epsilon}}{\log\log\frac{\Lambda\tau}{\epsilon}}\right),
\end{equation}
depending on the Hamiltonian. This is shown in detail in \cref{sec:derivations}, \cref{coroll:total_error_complexity}, which makes use of \cref{lemma:error_bound,lemma:error_bound_eps_L,coroll:epsilon_complexity,coroll:modified_complexity}.

\subsection{Insertion strategy}
The notion of partially included orders together with an expression for the error bound allows us to start from any given expansion $\vec{L}$ and determine which $L_k$ should be increased by 1 -- i.e. which additional gate should be included -- to give the quickest decrease of the error bound. Specifically, it is the $k$ which maximizes the expression
\begin{equation}\label{eq:insert_condition}
    \sum_{\nu \geq k}\,\frac{t^\nu}{\nu!}\,\alpha_{1+L_k}\prod_{\substack{j\neq k\\1\leq j \leq\nu}}\left(\sum_{i=1}^{L_j}\alpha_i \right).
\end{equation}
Starting from $\vec{L} = \vec{0}$, repeatedly adding terms that maximize \labelcref{eq:insert_condition} results in a greedy algorithm for decreasing the error bound, which we used to iteratively construct circuits.

\section{Results} \label{sec:results}

We first observe that for Hamiltonians with evenly distributed magnitudes $\alpha_\ell$, the only benefit of our modification is the finer control over the total gate count. By construction, whenever $C_{\vec{L}} = \nu L$, $\nu\in\mathds{N}$, our protocol and the method used in~\cite{berry2015taylor} yield identical results.

\begin{table}[b!]
    \caption{Molecules used in our calculations with their molecular formula, PubChem Compound ID (CID), number of qubits (excluding ancillas), and number of terms $L$.}
    \label{tbl:molecules}
    \vspace{1.5ex}
    \small\centering
    \setlength\tabcolsep{.8em}
    \begin{tabular}{lS[table-format=8.0]S[table-format=2.0]S[table-format=7.0]}\toprule
        {Formula} & {CID} & {Qubits} & {$L$} \\ \midrule
        HO & 157350 & 12 & 631\\
        HF & 16211014 & 12 & 631\\
        HN & 5460607 & 12 & 631\\
        LiH & 62714 & 12 & 631\\
        BH & 6397184 & 12 & 631\\
        BeH$_2$ & 139073 & 14 & 666\\
        CH$_2$ & 123164 & 14 & 1086\\
        NH$_2$ & 123329 & 14 & 1086\\
        BH$_2$ & 139760 & 14 & 1086\\
        H$_2$O & 962 & 14 & 1086\\
        BH$_3$ & 6331 & 16 & 1953\\
        CH$_3$ & 3034819 & 16 & 1969\\
        NH$_3$ & 222 & 16 & 2929\\
        CH$_4$ & 297 & 18 & 6892\\
        O$_2$ & 977 & 20 & 2239\\
        N$_2$ & 947 & 20 & 2951\\
        NO & 145068 & 20 & 4427\\
        CN & 5359238 & 20 & 5835\\
        BeO & 14775 & 20 & 5851\\
        LiF & 224478 & 20 & 5851\\
        CO & 281 & 20 & 5851\\
        BN & 66227 & 20 & 5851\\
        LiOH & 3939 & 22 & 8734\\
        HBO & 518615 & 22 & 8758\\
        HCN & 768 & 22 & 8758\\
        HOF & 123334 & 22 & 12070\\
        CHO & 123370 & 22 & 12070\\
        CHF & 186213 & 22 & 12074\\
        HNO & 945 & 22 & 12078\\
        H$_2$NO & 5460582 & 24 & 9257\\
        CH$_2$O & 712 & 24 & 9257\\
        NH$_2$F & 139987 & 24 & 15673\\
        CH$_2$F & 138041 & 24 & 15681\\
        CH$_3$F & 11638 & 26 & 18600\\
        CH$_3$Li & 2724049 & 26 & 19548\\
        H$_3$NO & 787 & 26 & 22080\\
        OCH$_3$ & 123146 & 26 & 39392\\
        LiBH$_4$ & 4148881 & 28 & 27473\\
        CH$_3$OH & 887 & 28 & 30419\\
        C$_4$H$_8$O$_2$ & 8857 & 76 & 1614647\\
        C$_8$H$_6$ & 12302244 & 92 & 1897809\\ \bottomrule
    \end{tabular}
\end{table}

\begin{figure}[!p]
    \centering
    \input{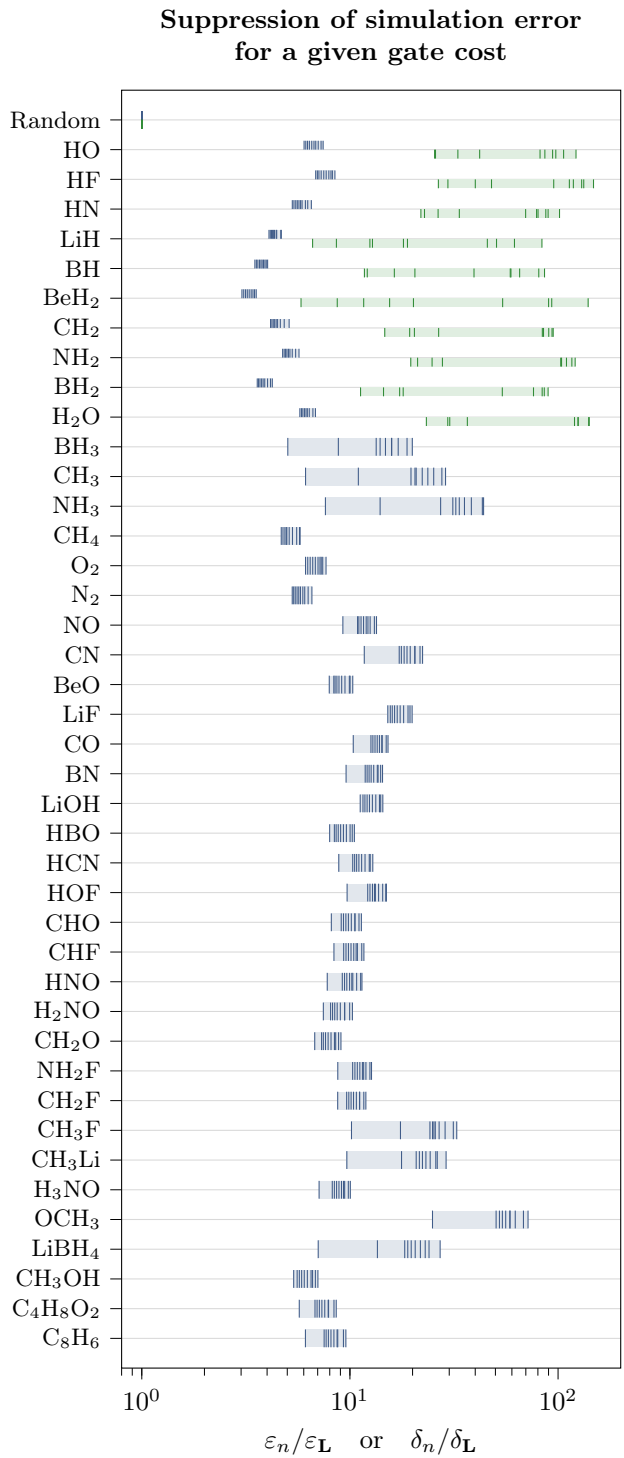}
    \caption{Ratio of the errors obtained without, versus with, our modification for different molecules at identical implementation cost, using a time step of $t_\infty$. Errors were evaluated at cost values $C_{\vec{L}} = (1\ldots 10)L$. Each vertical line represents the ratio of errors at some cost $C_{\vec{L}}$, the shaded areas indicate the range from the smallest to the largest data point. The advantage of our modified algorithm therefore increases left-to-right. Top-bottom split data indicates ratios of error bounds $\varepsilon_n/\varepsilon_{\vec{L}}$ at the top (marked in blue) and ratios of numerically obtained errors $\delta_n/\delta_{\vec{L}}$ at the bottom (marked in green). If there is no split, the blue lines represent error bounds only.}
    \label{fig:error_advantage}
\end{figure}
\begin{figure}[!p]
    \centering
    \input{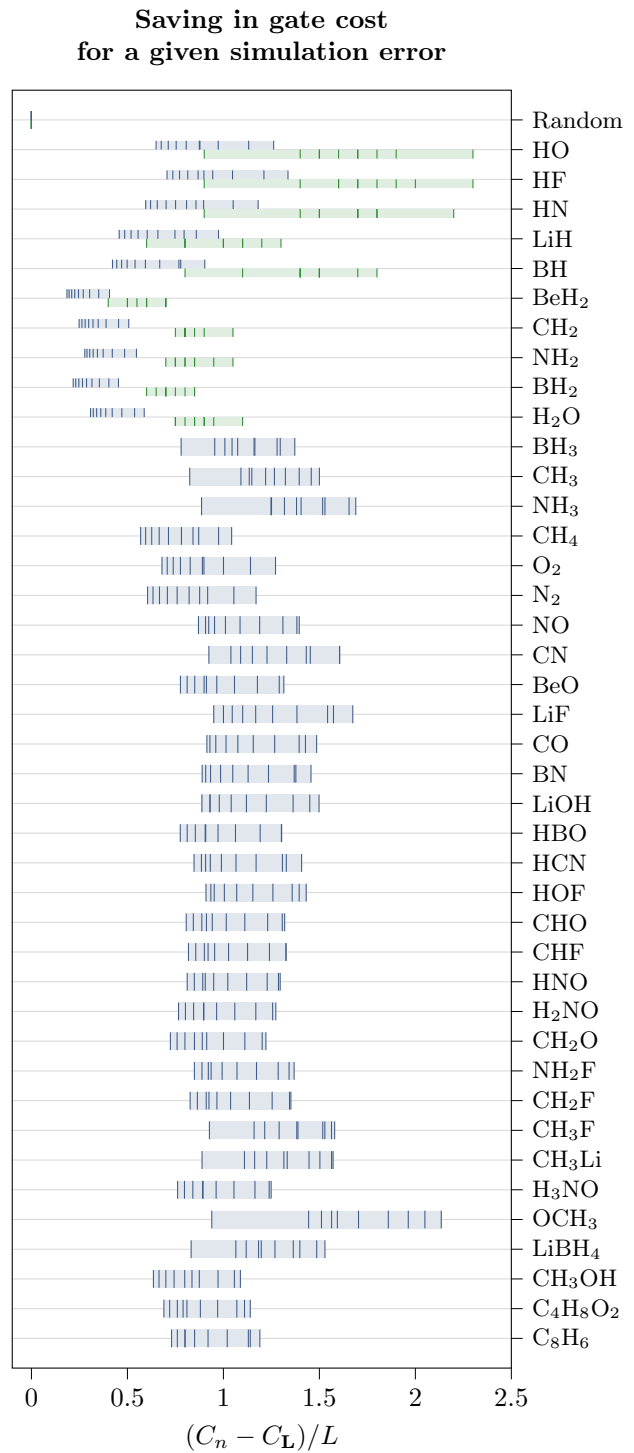}
    \caption{Difference between the cost of full expansions $C_n=nL$ to order $n=1\ldots10$ and the cost of an iteratively constructed circuit $C_{\vec{L}}$ to arrive at the same error, normalized to the cost of one full order $L$, for each molecule. Each vertical line represents the difference at some value of $n$, the shaded areas indicate the range from the smallest to the largest data point. The advantage of our modified algorithm therefore increases left-to-right. The time step size is $t_\infty$. Top-bottom split data indicates differences for error bounds at the top (marked in blue) and for numerically obtained errors on the bottom (marked in green). If there is no split, the blue lines represent error bounds only.}
    \label{fig:cost_advantage}
\end{figure}

We may expect our modification to be advantageous whenever the magnitudes of $\alpha_\ell$ vary over several orders of magnitude because this allows terms in low orders containing small $\alpha_\ell$ to be smaller than terms in higher orders containing large $\alpha_\ell$. Such magnitude distributions are often found in electronic structure Hamiltonians for molecules~\cite{helgaker2014molecular}. Because the efficiency of our method depends critically on the specific amplitudes in the Hamiltonian, analytical results are hard to obtain. Therefore we resort to a numerical study comparing the accuracy of the modification to the method in~\cite{berry2015taylor} for a group of molecules listed in \cref{tbl:molecules}.

\begin{figure}[!b]
    \centering
    \input{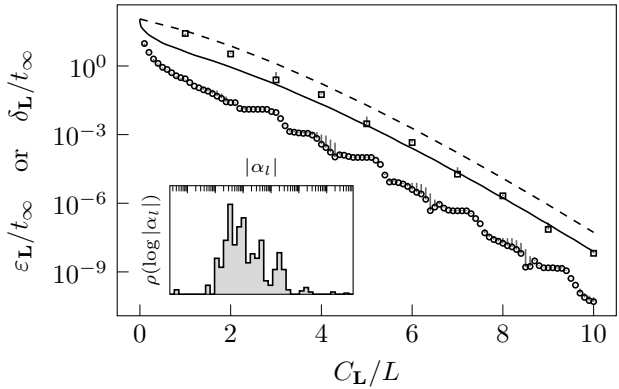}
    \vspace{-2ex}
    \caption{Accuracy of the Taylor expansion for the electronic Hamiltonian of hydrogen fluoride (HF), at time step size $t_\infty$, in terms of the error per unit time vs the circuit cost $C_{\vec{L}}$ as defined in \cref{eq:c} per cost of a full order. Lines are the error bounds $\varepsilon_{\vec{L}}$ for the unmodified~\ref{line:full_bound_hf} and modified~\ref{line:opt_bound_hf} circuit. Squares~\ref{line:full_amp} are the numerically obtained errors $\delta_{\vec{L}}$ for fully expanded orders, circles~\ref{line:opt_amp} analogous for partial orders. The vertical gray bars point to where the error would be if we could implement $U_{\vec{L}}$ without the amplification step. The inset shows the distribution $\rho$ of the logarithms of weights in the Hamiltonian $\log|\alpha_\ell|$.}
    \vspace{2ex}
    \label{fig:hf}
\end{figure}
\begin{figure}[!b]
    \centering
    \input{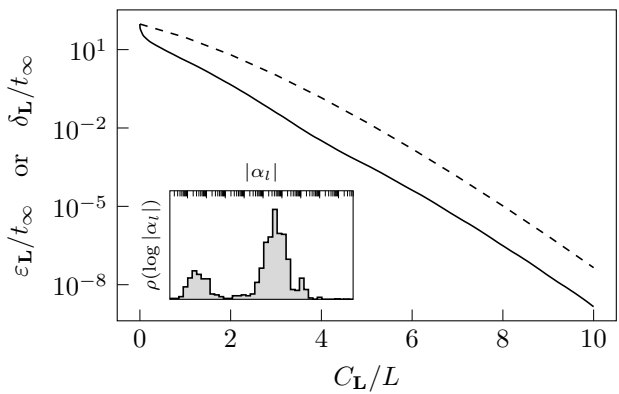}
    \vspace{-2ex}
    \caption{Identical plot to~\cref{fig:hf} showing only the convergence of error bounds for ammonia (NH$_3$). Notice the increasing distance between the unmodified~\ref{line:full_bound_nh3} and modified~\ref{line:opt_bound_nh3} variants of the algorithm between the orders~2 and 4, which is not present in~\cref{fig:hf}. It is caused by the two distinct clusters in the distribution of the logarithmic weights in the Hamiltonian visible in the inset.}
    \label{fig:nh3}
\end{figure}

The Hamiltonians for these molecules were obtained using OpenFermion~\cite{mcclean2017openfermion} and PySCF~\cite{pyscf}, with the basis set STO-3G~\cite{hehre1969sto}, and geometry data retrieved from PubChem~\cite{pubchem} and the NIST Computational Chemistry Comparison and Benchmark Database~\cite{cccbdb}. Mapping from second quantization to spin operators was done using the Jordan-Wigner transformation~\cite{jordan-wigner1928transform}. To showcase that our method yields improvements regardless of the basis set used, we also performed calculations using the cc-pVDZ, cc-pVTZ, and cc-pVQZ basis sets~\cite{dunning1989gaussian,prascher2011gaussian} for H$_2$ and LiH. The results are shown in~\cref{sec:basis_sets}.

In addition to the listed molecules, we also replaced the coefficients of the Hamiltonian for LiH with random numbers from a normal distribution with $\mu=1$ and $\sigma = 0.1$, to show the vanishing effect of our modification whenever all weights are similar. These results are labeled \emph{Random}.

\Cref{fig:hf} shows the error bounds as well as the numerically evaluated exact errors per unit time for hydrogen fluoride. Compared to the expansion to full orders, we see that our modification leads to a much quicker decrease of the error bound as well as the exact error in the range $0 < C_{\vec{L}} < L$, followed by very similar convergence for $C_{\vec{L}} > L$. This pattern is consistent with the convergence we observed for most other molecules we calculated.

Some compounds in our set -- namely BH$_3$, CH$_3$, NH$_3$, CH$_3$F, CH$_3$Li, OCH$_3$, and LiBH$_4$ -- show a slightly different behavior, where the ratio of error bounds using the modified and unmodified versions increases once more later in the iteration. \Cref{fig:nh3} illustrates this using NH$_3$ as an example. The delayed convergence is caused by a distinct second peak in the distribution of the logarithms of weights in the Hamiltonian $\log|\alpha_\ell|$, present in the mentioned molecules. The rest of our set shows distributions rather similar to that of HF in the inset in~\cref{fig:hf}. This disparity is also visible in the summarized results in~\cref{fig:error_advantage,fig:cost_advantage}, where the spread for the mentioned molecules (especially for the error in~\cref{fig:error_advantage}) is much greater than for the rest.

To summarize the results for all molecules, we obtained the ratio of the errors of the original and the modified version at cost values $C_{\vec{L}} = nL$, with $ n=1\ldots 10$, where the time step was set to $t_\infty$ for each respective molecule. The results are depicted in \cref{fig:error_advantage}. Across the listed molecules, our modification consistently yields errors roughly one order of magnitude lower than the unmodified method at equivalent costs, with some ratios as low as 3 and some as high as 100.

We also compared the cost required to obtain a certain error threshold. To this end, the errors $\delta_n$ of the expansions to full orders $n$ were calculated, and the cost $C_{\vec{L}}$ of the modified version to yield the same error was recorded. The results are depicted in \cref{fig:cost_advantage}. Using the modified method leads to saving approximately one order in most cases, i.e. the accuracy obtained by expanding $n$ full orders can be produced with a cost of $C_{\vec{L}} = (n-1) L$.

Our results show no strong correlation with neither the number of qubits in the Hamiltonian nor the number of terms $L$. Therefore we presume that these properties will also hold for other chemical Hamiltonians obtained in the same way.

\section{Conclusion} \label{sec:conclusion}
We demonstrated that a natural extension of the method proposed in~\cite{berry2015taylor} can lead to noticeable improvements in the convergence of the approximation when used for electronic structure Hamiltonians of molecules. The asymptotic behavior is equivalent; however, minimizing the required number of gates will be important for implementations on actual quantum hardware.

Our modification does not need the introduction of any new subroutines. It only rearranges the gate construction to facilitate quicker convergence of both the error bound and the actual error.

Due to the lack of analytic relations between the amplitudes in the Hamiltonians we investigated, only numeric results are available. However, because of the relatively large sample size of molecules we considered, it stands to reason that this behavior will generalize to a large portion of electronic structure Hamiltonians.

As mentioned in the introduction, in light of other methods like qubitization and quantum signal processing, the truncated Taylor series may be of particular relevance for diagonally dominant and time-dependent Hamiltonians. Investigating the suitability of our modification to such problems would therefore be an interesting question for future work.

Furthermore, combining our proposed adaptations with the improvements by~\citet{novo2016improved}, who add an additional correction step to the method, could also be worth exploring.

\begin{acknowledgments}
The authors would like to acknowledge the use of the University of Oxford Advanced Research Computing (ARC) facility (\textsc{doi}\,{\small\doi{10.5281/zenodo.22558}}) in carrying out this work. ETC was supported by the EPSRC (grant no. EP/M024261/1). SCB acknowledges support from the EU Flagship project AQTION, the NQIT Hub (EP/M013243/1), and the QCS Hub (EP/T001062/1). The authors also thank Bálint Koczor and Sam McArdle for useful discussions, as well as Yingkai Ouyang for helpful comments on the manuscript.
\end{acknowledgments}

\interlinepenalty=8000

\bibliography{references}

\onecolumngrid
\appendix

\section{Effect of basis set choice} \label{sec:basis_sets}
To ensure the effect we observed also holds for larger basis sets, we additionally performed calculations for H$_2$ and LiH with cc-pVDZ, cc-pVTZ, and cc-pVQZ sets.

\begin{figure}[hbtp]
\begin{minipage}[t]{.48\linewidth}\vspace{0pt}
    \centering
    % This file was created by tikzplotlib v0.9.4.
\begin{tikzpicture}

\definecolor{color0}{rgb}{0.250980392156863,0.36078431372549,0.537254901960784}
\definecolor{color1}{rgb}{0.887647058823529,0.904117647058823,0.930588235294118}

\begin{axis}[
height=.25\textheight,
width=.91\linewidth,
log basis x={10},
tick align=outside,
tick pos=left,
title style={align=center},
x grid style={white!69.0196078431373!black},
xlabel={\(\displaystyle \varepsilon_n/\varepsilon_{\vec{L}}\) \ \ or \ \ \(\displaystyle \delta_n/\delta_{\vec{L}}\)},
xmin=0.8, xmax=200,
xmode=log,
xtick style={color=black},
y dir=reverse,
y grid style={white!85!black},
y tick label style={font=\small},
ymajorgrids,
ymin=-1, ymax=7,
ytick style={color=black},
ytick={0,1,2,3,4,5,6},
yticklabels={H\(\displaystyle _2\) STO-3G,H\(\displaystyle _2\) cc-pVDZ,H\(\displaystyle _2\) cc-pVTZ,H\(\displaystyle _2\) cc-pVQZ,LiH STO-3G,LiH cc-pVDZ,LiH cc-pVTZ}
]
\draw[draw=none,fill=color1] (axis cs:1,-0.3) rectangle (axis cs:1.08597006654364,0.3);
\draw[draw=none,fill=color1] (axis cs:4.07704886568654,0.7) rectangle (axis cs:4.73874953959613,1.3);
\draw[draw=none,fill=color1] (axis cs:4.30529511672119,1.7) rectangle (axis cs:5.16287090335212,2.3);
\draw[draw=none,fill=color1] (axis cs:4.01132267159154,2.7) rectangle (axis cs:4.87003046524566,3.3);
\draw[draw=none,fill=color1] (axis cs:3.94151439131919,3.7) rectangle (axis cs:4.49856003149725,4.3);
\draw[draw=none,fill=color1] (axis cs:6.20639038870579,4.7) rectangle (axis cs:7.49634636351826,5.3);
\draw[draw=none,fill=color1] (axis cs:6.2820100312293,5.7) rectangle (axis cs:7.67246470837748,6.3);
\path [draw=color0]
(axis cs:1.08597006654364,-0.3)
--(axis cs:1.08597006654364,0.3);

\path [draw=color0]
(axis cs:1.01372404536149,-0.3)
--(axis cs:1.01372404536149,0.3);

\path [draw=color0]
(axis cs:1,-0.3)
--(axis cs:1,0.3);

\path [draw=color0]
(axis cs:1,-0.3)
--(axis cs:1,0.3);

\path [draw=color0]
(axis cs:1,-0.3)
--(axis cs:1,0.3);

\path [draw=color0]
(axis cs:1,-0.3)
--(axis cs:1,0.3);

\path [draw=color0]
(axis cs:1,-0.3)
--(axis cs:1,0.3);

\path [draw=color0]
(axis cs:1,-0.3)
--(axis cs:1,0.3);

\path [draw=color0]
(axis cs:1,-0.3)
--(axis cs:1,0.3);

\path [draw=color0]
(axis cs:1,-0.3)
--(axis cs:1,0.3);

\path [draw=color0]
(axis cs:4.49885127666608,0.7)
--(axis cs:4.49885127666608,1.3);

\path [draw=color0]
(axis cs:4.73874953959613,0.7)
--(axis cs:4.73874953959613,1.3);

\path [draw=color0]
(axis cs:4.54369167700147,0.7)
--(axis cs:4.54369167700147,1.3);

\path [draw=color0]
(axis cs:4.40465890003587,0.7)
--(axis cs:4.40465890003587,1.3);

\path [draw=color0]
(axis cs:4.33300935118244,0.7)
--(axis cs:4.33300935118244,1.3);

\path [draw=color0]
(axis cs:4.27362781994548,0.7)
--(axis cs:4.27362781994548,1.3);

\path [draw=color0]
(axis cs:4.21797284485053,0.7)
--(axis cs:4.21797284485053,1.3);

\path [draw=color0]
(axis cs:4.16233156839849,0.7)
--(axis cs:4.16233156839849,1.3);

\path [draw=color0]
(axis cs:4.11525903508118,0.7)
--(axis cs:4.11525903508118,1.3);

\path [draw=color0]
(axis cs:4.07704886568654,0.7)
--(axis cs:4.07704886568654,1.3);

\path [draw=color0]
(axis cs:4.56615406153256,1.7)
--(axis cs:4.56615406153256,2.3);

\path [draw=color0]
(axis cs:5.14848923683138,1.7)
--(axis cs:5.14848923683138,2.3);

\path [draw=color0]
(axis cs:5.16287090335212,1.7)
--(axis cs:5.16287090335212,2.3);

\path [draw=color0]
(axis cs:4.97672770645665,1.7)
--(axis cs:4.97672770645665,2.3);

\path [draw=color0]
(axis cs:4.80819212041967,1.7)
--(axis cs:4.80819212041967,2.3);

\path [draw=color0]
(axis cs:4.6710783399042,1.7)
--(axis cs:4.6710783399042,2.3);

\path [draw=color0]
(axis cs:4.55759697247947,1.7)
--(axis cs:4.55759697247947,2.3);

\path [draw=color0]
(axis cs:4.46263753621292,1.7)
--(axis cs:4.46263753621292,2.3);

\path [draw=color0]
(axis cs:4.37876463233357,1.7)
--(axis cs:4.37876463233357,2.3);

\path [draw=color0]
(axis cs:4.30529511672119,1.7)
--(axis cs:4.30529511672119,2.3);

\path [draw=color0]
(axis cs:4.26274253555316,2.7)
--(axis cs:4.26274253555316,3.3);

\path [draw=color0]
(axis cs:4.87003046524566,2.7)
--(axis cs:4.87003046524566,3.3);

\path [draw=color0]
(axis cs:4.846710183039,2.7)
--(axis cs:4.846710183039,3.3);

\path [draw=color0]
(axis cs:4.67351115986682,2.7)
--(axis cs:4.67351115986682,3.3);

\path [draw=color0]
(axis cs:4.50676511562117,2.7)
--(axis cs:4.50676511562117,3.3);

\path [draw=color0]
(axis cs:4.37401464415335,2.7)
--(axis cs:4.37401464415335,3.3);

\path [draw=color0]
(axis cs:4.26331796178155,2.7)
--(axis cs:4.26331796178155,3.3);

\path [draw=color0]
(axis cs:4.16844917205035,2.7)
--(axis cs:4.16844917205035,3.3);

\path [draw=color0]
(axis cs:4.08552099588359,2.7)
--(axis cs:4.08552099588359,3.3);

\path [draw=color0]
(axis cs:4.01132267159154,2.7)
--(axis cs:4.01132267159154,3.3);

\path [draw=color0]
(axis cs:4.25916146675998,3.7)
--(axis cs:4.25916146675998,4.3);

\path [draw=color0]
(axis cs:4.49856003149725,3.7)
--(axis cs:4.49856003149725,4.3);

\path [draw=color0]
(axis cs:4.46618122916808,3.7)
--(axis cs:4.46618122916808,4.3);

\path [draw=color0]
(axis cs:4.25002978289656,3.7)
--(axis cs:4.25002978289656,4.3);

\path [draw=color0]
(axis cs:4.15838277445898,3.7)
--(axis cs:4.15838277445898,4.3);

\path [draw=color0]
(axis cs:4.13026428684679,3.7)
--(axis cs:4.13026428684679,4.3);

\path [draw=color0]
(axis cs:4.10073047300003,3.7)
--(axis cs:4.10073047300003,4.3);

\path [draw=color0]
(axis cs:4.05394748858842,3.7)
--(axis cs:4.05394748858842,4.3);

\path [draw=color0]
(axis cs:4.0006753415361,3.7)
--(axis cs:4.0006753415361,4.3);

\path [draw=color0]
(axis cs:3.94151439131919,3.7)
--(axis cs:3.94151439131919,4.3);

\path [draw=color0]
(axis cs:6.50262756314478,4.7)
--(axis cs:6.50262756314478,5.3);

\path [draw=color0]
(axis cs:7.49634636351826,4.7)
--(axis cs:7.49634636351826,5.3);

\path [draw=color0]
(axis cs:7.40121163846627,4.7)
--(axis cs:7.40121163846627,5.3);

\path [draw=color0]
(axis cs:7.13203706798774,4.7)
--(axis cs:7.13203706798774,5.3);

\path [draw=color0]
(axis cs:6.90402719595343,4.7)
--(axis cs:6.90402719595343,5.3);

\path [draw=color0]
(axis cs:6.72968328559108,4.7)
--(axis cs:6.72968328559108,5.3);

\path [draw=color0]
(axis cs:6.58097529216399,4.7)
--(axis cs:6.58097529216399,5.3);

\path [draw=color0]
(axis cs:6.44548895871298,4.7)
--(axis cs:6.44548895871298,5.3);

\path [draw=color0]
(axis cs:6.31895777446782,4.7)
--(axis cs:6.31895777446782,5.3);

\path [draw=color0]
(axis cs:6.20639038870579,4.7)
--(axis cs:6.20639038870579,5.3);

\path [draw=color0]
(axis cs:6.46387090933878,5.7)
--(axis cs:6.46387090933878,6.3);

\path [draw=color0]
(axis cs:7.67246470837748,5.7)
--(axis cs:7.67246470837748,6.3);

\path [draw=color0]
(axis cs:7.65330650762276,5.7)
--(axis cs:7.65330650762276,6.3);

\path [draw=color0]
(axis cs:7.33617638047875,5.7)
--(axis cs:7.33617638047875,6.3);

\path [draw=color0]
(axis cs:7.05594523418256,5.7)
--(axis cs:7.05594523418256,6.3);

\path [draw=color0]
(axis cs:6.85020458735965,5.7)
--(axis cs:6.85020458735965,6.3);

\path [draw=color0]
(axis cs:6.68123608543701,5.7)
--(axis cs:6.68123608543701,6.3);

\path [draw=color0]
(axis cs:6.53359321253853,5.7)
--(axis cs:6.53359321253853,6.3);

\path [draw=color0]
(axis cs:6.40131717477078,5.7)
--(axis cs:6.40131717477078,6.3);

\path [draw=color0]
(axis cs:6.2820100312293,5.7)
--(axis cs:6.2820100312293,6.3);

\end{axis}

\end{tikzpicture}\vspace{-1.95pt} % need to adjust the caption offset because the x labels have different heights
    \caption{Ratio of the error bounds obtained without, versus with, our modification for different molecules and basis sets, using a time step of $t_\infty$. Errors were evaluated at each cost value $C_{\vec{L}} = (1\ldots 10)L$. Each line represents the ratio of error bounds $\varepsilon_n/\varepsilon_{\vec{L}}$ at some cost $C_{\vec{L}}$. The advantage of our modified algorithm therefore increases left-to-right.}
    \label{fig:error_advantage_basis_sets}
\end{minipage}
\hfill
\begin{minipage}[t]{.48\linewidth}\vspace{0pt}
    \centering
    % This file was created by tikzplotlib v0.9.4.
\begin{tikzpicture}

\definecolor{color0}{rgb}{0.250980392156863,0.36078431372549,0.537254901960784}
\definecolor{color1}{rgb}{0.887647058823529,0.904117647058823,0.930588235294118}

\begin{axis}[
height=.25\textheight,
width=.91\linewidth,
tick align=outside,
title style={align=center},
x grid style={white!69.0196078431373!black},
xlabel={\(\displaystyle (C_n - C_{\vec{L}})/L\)},
xmin=-0.1, xmax=2.5,
xtick pos=left,
xtick style={color=black},
y dir=reverse,
y grid style={white!85!black},
y tick label style={font=\small},
ymajorgrids,
ymin=-1, ymax=7,
ytick pos=right,
ytick style={color=black},
ytick={0,1,2,3,4,5,6},
yticklabels={H\(\displaystyle _2\) STO-3G,H\(\displaystyle _2\) cc-pVDZ,H\(\displaystyle _2\) cc-pVTZ,H\(\displaystyle _2\) cc-pVQZ,LiH STO-3G,LiH cc-pVDZ,LiH cc-pVTZ}
]
\draw[draw=none,fill=color1] (axis cs:0,-0.3) rectangle (axis cs:0.0714285714285714,0.3);
\draw[draw=none,fill=color1] (axis cs:0.470169491525424,0.7) rectangle (axis cs:0.98,1.3);
\draw[draw=none,fill=color1] (axis cs:0.520001275998469,1.7) rectangle (axis cs:0.990004678661052,2.3);
\draw[draw=none,fill=color1] (axis cs:0.5,2.7) rectangle (axis cs:0.950000152191095,3.3);
\draw[draw=none,fill=color1] (axis cs:0.44,3.7) rectangle (axis cs:0.94,4.3);
\draw[draw=none,fill=color1] (axis cs:0.650005796952091,4.7) rectangle (axis cs:1.20000772926946,5.3);
\draw[draw=none,fill=color1] (axis cs:0.660000255522946,5.7) rectangle (axis cs:1.19000001825164,6.3);
\path [draw=color0]
(axis cs:0.0714285714285714,-0.3)
--(axis cs:0.0714285714285714,0.3);

\path [draw=color0]
(axis cs:0,-0.3)
--(axis cs:0,0.3);

\path [draw=color0]
(axis cs:0,-0.3)
--(axis cs:0,0.3);

\path [draw=color0]
(axis cs:0,-0.3)
--(axis cs:0,0.3);

\path [draw=color0]
(axis cs:0,-0.3)
--(axis cs:0,0.3);

\path [draw=color0]
(axis cs:0,-0.3)
--(axis cs:0,0.3);

\path [draw=color0]
(axis cs:0,-0.3)
--(axis cs:0,0.3);

\path [draw=color0]
(axis cs:0,-0.3)
--(axis cs:0,0.3);

\path [draw=color0]
(axis cs:0,-0.3)
--(axis cs:0,0.3);

\path [draw=color0]
(axis cs:0,-0.3)
--(axis cs:0,0.3);

\path [draw=color0]
(axis cs:0.850169491525424,0.7)
--(axis cs:0.850169491525424,1.3);

\path [draw=color0]
(axis cs:0.98,0.7)
--(axis cs:0.98,1.3);

\path [draw=color0]
(axis cs:0.86,0.7)
--(axis cs:0.86,1.3);

\path [draw=color0]
(axis cs:0.74,0.7)
--(axis cs:0.74,1.3);

\path [draw=color0]
(axis cs:0.66,0.7)
--(axis cs:0.66,1.3);

\path [draw=color0]
(axis cs:0.6,0.7)
--(axis cs:0.6,1.3);

\path [draw=color0]
(axis cs:0.56,0.7)
--(axis cs:0.56,1.3);

\path [draw=color0]
(axis cs:0.52,0.7)
--(axis cs:0.52,1.3);

\path [draw=color0]
(axis cs:0.490169491525424,0.7)
--(axis cs:0.490169491525424,1.3);

\path [draw=color0]
(axis cs:0.470169491525424,0.7)
--(axis cs:0.470169491525424,1.3);

\path [draw=color0]
(axis cs:0.850001063332057,1.7)
--(axis cs:0.850001063332057,2.3);

\path [draw=color0]
(axis cs:0.990004678661052,1.7)
--(axis cs:0.990004678661052,2.3);

\path [draw=color0]
(axis cs:0.910004891327464,1.7)
--(axis cs:0.910004891327464,2.3);

\path [draw=color0]
(axis cs:0.800003189996172,1.7)
--(axis cs:0.800003189996172,2.3);

\path [draw=color0]
(axis cs:0.720003402662583,1.7)
--(axis cs:0.720003402662583,2.3);

\path [draw=color0]
(axis cs:0.660004891327464,1.7)
--(axis cs:0.660004891327464,2.3);

\path [draw=color0]
(axis cs:0.610001701331292,1.7)
--(axis cs:0.610001701331292,2.3);

\path [draw=color0]
(axis cs:0.580005103993875,1.7)
--(axis cs:0.580005103993875,2.3);

\path [draw=color0]
(axis cs:0.550003189996172,1.7)
--(axis cs:0.550003189996172,2.3);

\path [draw=color0]
(axis cs:0.520001275998469,1.7)
--(axis cs:0.520001275998469,2.3);

\path [draw=color0]
(axis cs:0.790000131898949,2.7)
--(axis cs:0.790000131898949,3.3);

\path [draw=color0]
(axis cs:0.950000152191095,2.7)
--(axis cs:0.950000152191095,3.3);

\path [draw=color0]
(axis cs:0.870000142045022,2.7)
--(axis cs:0.870000142045022,3.3);

\path [draw=color0]
(axis cs:0.770000192775387,2.7)
--(axis cs:0.770000192775387,3.3);

\path [draw=color0]
(axis cs:0.690000182629314,2.7)
--(axis cs:0.690000182629314,3.3);

\path [draw=color0]
(axis cs:0.630000111606803,2.7)
--(axis cs:0.630000111606803,3.3);

\path [draw=color0]
(axis cs:0.590000233359679,2.7)
--(axis cs:0.590000233359679,3.3);

\path [draw=color0]
(axis cs:0.55000010146073,2.7)
--(axis cs:0.55000010146073,3.3);

\path [draw=color0]
(axis cs:0.520000192775387,2.7)
--(axis cs:0.520000192775387,3.3);

\path [draw=color0]
(axis cs:0.5,2.7)
--(axis cs:0.5,3.3);

\path [draw=color0]
(axis cs:0.79,3.7)
--(axis cs:0.79,4.3);

\path [draw=color0]
(axis cs:0.94,3.7)
--(axis cs:0.94,4.3);

\path [draw=color0]
(axis cs:0.82,3.7)
--(axis cs:0.82,4.3);

\path [draw=color0]
(axis cs:0.72,3.7)
--(axis cs:0.72,4.3);

\path [draw=color0]
(axis cs:0.64,3.7)
--(axis cs:0.64,4.3);

\path [draw=color0]
(axis cs:0.58,3.7)
--(axis cs:0.58,4.3);

\path [draw=color0]
(axis cs:0.53,3.7)
--(axis cs:0.53,4.3);

\path [draw=color0]
(axis cs:0.5,3.7)
--(axis cs:0.5,4.3);

\path [draw=color0]
(axis cs:0.46,3.7)
--(axis cs:0.46,4.3);

\path [draw=color0]
(axis cs:0.44,3.7)
--(axis cs:0.44,4.3);

\path [draw=color0]
(axis cs:0.910005539309776,4.7)
--(axis cs:0.910005539309776,5.3);

\path [draw=color0]
(axis cs:1.20000772926946,4.7)
--(axis cs:1.20000772926946,5.3);

\path [draw=color0]
(axis cs:1.13000631223672,4.7)
--(axis cs:1.13000631223672,5.3);

\path [draw=color0]
(axis cs:1,4.7)
--(axis cs:1,5.3);

\path [draw=color0]
(axis cs:0.890012495652286,4.7)
--(axis cs:0.890012495652286,5.3);

\path [draw=color0]
(axis cs:0.820011078619552,4.7)
--(axis cs:0.820011078619552,5.3);

\path [draw=color0]
(axis cs:0.760006183415564,4.7)
--(axis cs:0.760006183415564,5.3);

\path [draw=color0]
(axis cs:0.720007213984825,4.7)
--(axis cs:0.720007213984825,5.3);

\path [draw=color0]
(axis cs:0.680008244554086,4.7)
--(axis cs:0.680008244554086,5.3);

\path [draw=color0]
(axis cs:0.650005796952091,4.7)
--(axis cs:0.650005796952091,5.3);

\path [draw=color0]
(axis cs:0.880000036503278,5.7)
--(axis cs:0.880000036503278,6.3);

\path [draw=color0]
(axis cs:1.19000001825164,5.7)
--(axis cs:1.19000001825164,6.3);

\path [draw=color0]
(axis cs:1.14000010950983,5.7)
--(axis cs:1.14000010950983,6.3);

\path [draw=color0]
(axis cs:1.01000007300656,5.7)
--(axis cs:1.01000007300656,6.3);

\path [draw=color0]
(axis cs:0.910000255522946,5.7)
--(axis cs:0.910000255522946,6.3);

\path [draw=color0]
(axis cs:0.830000127761473,5.7)
--(axis cs:0.830000127761473,6.3);

\path [draw=color0]
(axis cs:0.770000146013112,5.7)
--(axis cs:0.770000146013112,6.3);

\path [draw=color0]
(axis cs:0.730000310277863,5.7)
--(axis cs:0.730000310277863,6.3);

\path [draw=color0]
(axis cs:0.690000018251639,5.7)
--(axis cs:0.690000018251639,6.3);

\path [draw=color0]
(axis cs:0.660000255522946,5.7)
--(axis cs:0.660000255522946,6.3);

\end{axis}

\end{tikzpicture}
    \caption{Difference of the cost of full expansions $C_n=nL$ to order $n=1\ldots10$ and the cost of an iteratively constructed circuit $C_{\vec{L}}$ to arrive at the same error bound, normalized to the cost of one full order $L$, for each molecule and basis set. The advantage of our modified algorithm therefore increases left-to-right. The time step size is $t_\infty$.}
    \label{fig:cost_advantage_basis_sets}
\end{minipage}
\end{figure}

\Cref{fig:error_advantage_basis_sets,fig:cost_advantage_basis_sets} show that for the considered cases, larger basis sets seem to slightly enhance the advantage of our modification. As a special case, for H$_2$ our proposed method yields almost no improvement when using STO-3G, due to the very low number of only 15 terms. Apart from this outlier, the influence of the choice of basis set on our algorithm's performance seems small.

\pagebreak
\section{Proofs} \label{sec:derivations}
For completeness and convenience, we collect several of the results we used in this appendix, including some that may be well known.
\addtolength{\jot}{4pt}

\begin{lemma} \label{lemma:amp}
    The optimal choice for the number of amplification steps is $\nu = 1$, resulting in $s_{\vec{L}} = 2$.
\end{lemma}
\begin{proof}
    For unitary $U_{\vec{L}}$, the operator $\mathcal{Q}^\nu \mathcal{W}$, with $\mathcal{Q}$ as defined in \cref{eq:q}, would have the effect~\cite{berry2014improvement}
    \begin{align*}
        \mathcal{Q}^\nu \mathcal{W} \ket{0}\ket{\psi} = &\sin{[(2\nu+1)\sin^{-1}(s_{\vec{L}}^{-1})]} \ket{0} U_{\vec{L}} \ket{\psi}\\
        + &\cos{[(2\nu+1)\cos^{-1}(N)]} \ket{0^\mathsmaller{\perp}, \Phi}.
    \end{align*}
    For any given number of amplification steps $\nu$, the amplitude of the desired state $\ket{0} U_{\vec{L}} \ket{\psi}$ can be tuned to 1 by setting $t$ such that $s_{\vec{L}}$ fulfills
    \begin{equation}\label{eq:s_n_nu}
        s_{\vec{L}} = \sin{\left(\frac{\pi}{4\nu+2}\right)}^{-1} \sim \frac{4\nu + 2}{\pi}.
    \end{equation}

    For this argument it is sufficient to analyze the full expansion to order $n$. To find the optimal number $\nu$ we consider the operator $\mathcal{Q}^\nu\mathcal{W}$, which contains the most expensive operator $\mathcal{W}$ a total of $2\nu + 1$ times. Therefore the cost is approximately linear in $\nu$, meaning it is also linear in $s_n$. However, we know that
    \begin{equation} \label{eq:s_n_t}
        s_n = \sum_{k=0}^n \frac{t^k}{k!}\, \Bigg(\!\underbrace{\sum_{\ell = 0}^{L-1} \alpha_{\ell}}_{\coloneqq \Lambda}\!\Bigg)^k = \sum_{k=0}^n \frac{t^k \Lambda^k}{k!} \approx e^{\Lambda t},
    \end{equation}
    where the rightmost approximation holds for sufficiently large $n$. \Cref{eq:s_n_nu,eq:s_n_t} imply that the cost is exponential in $t$, indicating there is no benefit in amplifying more than once. Exact numerical evaluation shows that for $n\to\infty$, one or two amplification steps ($\nu \in \{1, 2\}$) yield approximately equivalent time-per-gate, but the smaller $t$ of $\nu=1$ leads to quicker convergence in $n$. Consequently, it is best to choose $\nu=1$. This choice forces $s_{\vec{L}}$ to satisfy
    \begin{equation*}
        s_{\vec{L}} = \sin\left(\frac{\pi}{6}\right)^{-1} = 2.
    \end{equation*}
\end{proof}

\begin{lemma} \label{lemma:action_A}
    The action of $\Pi\mathcal{A}$ on a product state $\ket{0}\ket{\psi}$ is given by~\cite{berry2015taylor}
    \begin{equation*}
        \Pi\mathcal{A} \ket{0}\ket{\psi} = \ket{0} \left(-\frac{4}{s_{\vec{L}}^3} U_{\vec{L}}^{\vphantom{\dagger}} U_{\vec{L}}^\dagger U_{\vec{L}}^{\vphantom{\dagger}} + \frac{3}{s_{\vec{L}}} U_{\vec{L}}\right) \ket{\psi}.
    \end{equation*}
\end{lemma}
\begin{proof}
    We can explicitly expand $R$ and use that $\Pi^2=\Pi$ as well as $\Pi \ket{0}\ket{\psi} = \ket{0}\ket{\psi}$ to find
    \begin{align*}
        \Pi\mathcal{A} \ket{0}\ket{\psi} &= -\Pi\mathcal{W}R\mathcal{W}^\dagger R \mathcal{W}\ket{0}\ket{\psi}\\
        &= -\Pi\mathcal{W} (2\Pi - \mathds{1}) \mathcal{W}^\dagger (2\Pi - \mathds{1}) \mathcal{W} \ket{0}\ket{\psi}\\
        &= (-4\Pi\mathcal{W}\Pi\mathcal{W}^\dagger\Pi\mathcal{W} + 2\Pi\mathcal{W}\mathcal{W}^\dagger\Pi\mathcal{W} + 2\Pi\mathcal{W}\Pi\mathcal{W}^\dagger\mathcal{W} - \Pi\mathcal{W}\mathcal{W}^\dagger\mathcal{W}) \ket{0}\ket{\psi}\\
        &= (-4\Pi\mathcal{W}\Pi\mathcal{W}^\dagger\Pi\mathcal{W} + 3 \Pi\mathcal{W})\ket{0}\ket{\psi}\\
        &= (-4\underbrace{\Pi\mathcal{W}\Pi}_{\mathclap{\frac{1}{s_{\vec{L}}}(\ketbra{0}{0}\otimes U_{\vec{L}})}}\Pi\mathcal{W}^\dagger\Pi\Pi\mathcal{W}\Pi + 3 \Pi\mathcal{W}\Pi)\ket{0}\ket{\psi}\\
        &= \ket{0} \left(-\frac{4}{s_{\vec{L}}^3} U_{\vec{L}}^{\vphantom{\dagger}} U_{\vec{L}}^\dagger U_{\vec{L}}^{\vphantom{\dagger}} + \frac{3}{s_{\vec{L}}} U_{\vec{L}}\right) \ket{\psi}
    \end{align*}
    as claimed.
\end{proof}

\pagebreak
\begin{lemma} \label{lemma:p_star}
    If used in $\mathcal{W}\ket{0}\ket{\psi}$, $\mathcal{P}^\star$ has the same effect as $\mathcal{P}$, i.e. $\mathcal{W}\ket{0}\ket{\psi} = (\mathcal{P}^\dagger \otimes \mathds{1})\,\mathcal{S}\,(\mathcal{P} \otimes \mathds{1})\ket{0}\ket{\psi} = ({\mathcal{P}^\star}^\dagger \otimes \mathds{1})\,\mathcal{S}\,(\mathcal{P}^\star \otimes \mathds{1})\ket{0}\ket{\psi}$
\end{lemma}
\begin{proof}
    $\mathcal{P}^\star$ on any of the $c_k$ registers has the action
    \begin{equation*}
        \mathcal{P}^\star \ket{0}_{\!c_k} = \frac{1}{\sqrt{\Lambda_k}} \sum_{\ell=0}^{L_k - 1} \sqrt{\alpha_\ell}\ket{\ell}_{\!c_k}
    \end{equation*}
    and on the $q$ register
    \begin{equation*}
        \mathcal{P}^\star \ket{0}_{\!q} = \frac{1}{\sqrt{N_q}} \sum_{k=0}^{\kappa} \sqrt{\frac{t^k}{k!}\prod_{j=1}^k \Lambda_j} \ket{k}_{\!q}
    \end{equation*}
    where $N_q = \sum_{k=0}^{\infty}\frac{t^k}{k!}\prod_{j=1}^k \Lambda_j$ and $\kappa$ is the largest nonzero index in $\vec{L}$. Therefore
    \begin{equation*}
        \mathcal{S}\,(\mathcal{P}^\star \otimes \mathds{1})\ket{0}\ket{\psi} = \frac{1}{\sqrt{N_q \prod_{k=1}^\kappa \Lambda_k}} \sum_{k=0}^\kappa \sqrt{\left(\prod_{j=1}^k \Lambda_j \right) \frac{t^k}{k!}}\ket{k}_{\!q} \bigotimes_{j=1}^k \left(\sum_{\ell=0}^{L_j - 1}\ket{\ell}_{\!c_j} \sqrt{\alpha_\ell} \tilde{h}_\ell \right) \bigotimes_{j=k+1}^\kappa \left(\sum_{\ell=0}^{L_j - 1}\ket{\ell}_{\!c_j} \sqrt{\alpha_\ell}\right) \ket{\psi}
    \end{equation*}
    Transforming back with ${\mathcal{P}^\star}^\dagger$ and projecting onto the ancilla $\ket{0}$ yields
    \begin{align*}
        \Pi ({\mathcal{P}^\star}^\dagger \otimes \mathds{1})\,\mathcal{S}\,(\mathcal{P}^\star \otimes \mathds{1}) &= \frac{\Pi}{N_q \prod_{k=1}^\kappa \Lambda_k} \sum_{k=0}^\kappa \left(\prod_{j=1}^k \Lambda_j \right) \frac{t^k}{k!}\ket{0}_{\!q} \bigotimes_{j=1}^k \left(\sum_{\ell=0}^{L_j - 1}\ket{0}_{\!c_j} \alpha_\ell \tilde{h}_\ell \right) \bigotimes_{j=k+1}^\kappa \underbrace{\left(\sum_{\ell=0}^{L_j - 1}\ket{0}_{\!c_j} \alpha_\ell\right)}_{\Lambda_j \ket{0}_{c_j}} \ket{\psi}\\
        &= \ket{0} \frac{1}{N_q \prod_{k=1}^\kappa \Lambda_k} \sum_{k=0}^\kappa \left(\prod_{j=1}^\kappa \Lambda_j \right) \frac{t^k}{k!} \prod_{j=1}^k \left(\sum_{\ell=0}^{L_j - 1} \alpha_\ell \tilde{h}_\ell \right) \ket{\psi}\\
        &= \frac{1}{N_q} \ket{0} \sum_{k=0}^\kappa \frac{t^k}{k!} \prod_{j=1}^k \left(\sum_{\ell=0}^{L_j - 1} \alpha_\ell \tilde{h}_\ell \right) \ket{\psi} = \frac{1}{N_q} \ket{0} U_{\vec{L}} \ket{\psi} = \Pi \mathcal{W}\ket{0}\ket{\psi}
    \end{align*}
    which is \cref{eq:W} and we see that $N_q = s_{\vec{L}}$ as defined in \cref{eq:s_L}.
\end{proof}

\begin{lemma} \label{lemma:delta_L}
    The error of a single time step $\delta_{\vec{L}}$ when using $t_\infty = \log(2)/\Lambda$ can be bounded by
    \begin{equation*}
        \delta_{\vec{L}}(t_\infty) \leq 2 - s_{\vec{L}}(t_\infty) \eqqcolon \varepsilon_{\vec{L}}.
    \end{equation*}
\end{lemma}
\begin{proof}
    For easier notation, we first consider fully expanded orders and define
    \begin{align*}
        \tilde{\Delta}_n(t) &\coloneqq U_n(t) - U(t)\\
        \varepsilon_n &\coloneqq s_\infty(t_\infty) - s_n(t_\infty) = 2 - s_n(t_\infty)
    \end{align*}
    and observe that
    \begin{align*}
        ||\tilde{\Delta}_n(t)|| &= \left|\left| \sum_{k=1}^n \frac{(-it)^k}{k!}\sum_{\ell_1, \ldots, \ell_k=0}^{L-1} \alpha_{\ell_1} \ldots \alpha_{\ell_k} h_{\ell_1} \ldots h_{\ell_k} - \sum_{k=1}^{\infty} \frac{(-it)^k}{k!} \sum_{\ell_1, \ldots, \ell_k}^L \alpha_{\ell_1} \ldots \alpha_{\ell_k} h_{\ell_1} \ldots h_{\ell_k} \right|\right|\\
        &= \left|\left| \sum_{k=n+1}^{\infty} \frac{(-it)^k}{k!} \sum_{\ell_1, \ldots, \ell_k=0}^{L-1} \alpha_{\ell_1} \ldots \alpha_{\ell_k} h_{\ell_1} \ldots h_{\ell_k} \right|\right|\\
        &\leq \sum_{k=n+1}^{\infty} \frac{t^k}{k!} \sum_{\ell_1, \ldots, \ell_k = 0}^{L-1} \alpha_{\ell_1} \ldots \alpha_{\ell_k} \underbrace{||h_{\ell_1}|| \ldots ||h_{\ell_k}||}_1\\
        &= \left(1 + \sum_{k=1}^{\infty} \frac{t^k}{k!} \sum_{\ell_1, \ldots, \ell_k = 0}^{L-1} \alpha_{\ell_1} \ldots \alpha_{\ell_k} \right) - \left(1 + \sum_{k=1}^{n} \frac{t^k}{k!}\sum_{\ell_1, \ldots, \ell_k = 0}^{L-1} \alpha_{\ell_1} \ldots \alpha_{\ell_k} \right)\\
        &= s_\infty(t) - s_n(t)
    \end{align*}
    which means
    \begin{equation*}
        ||\tilde{\Delta}_n(t_\infty)|| \leq s_\infty(t_\infty) - s_n(t_\infty) = \varepsilon_n.
    \end{equation*}
    Using these in our definition of $\Delta_n$ yields
    \begin{align*}
        -\Delta_{n}(t_\infty) &= \tilde{\mathcal{A}}_n(t_\infty) - U(t_\infty) \\
        &= \frac{3}{s_n(t_\infty)} U_n(t_\infty) - \frac{4}{s_n^3(t_\infty)} U_n^{\phantom{\dagger}}(t_\infty) U_n^\dagger(t_\infty) U_n^{\phantom{\dagger}}(t_\infty) - U(t_\infty)\\
        &= \underbrace{\frac{3}{2 - \varepsilon_n}}_{\mathclap{=\frac{3}{2} + \frac{3\varepsilon_n}{4} + \mathcal{O}(\varepsilon_n^2)}} (U + \tilde{\Delta}_n) - \underbrace{\frac{4}{(2-\varepsilon_n)^3}}_{\mathclap{=\frac{1}{2} + \frac{3\varepsilon_n}{4} + \mathcal{O}(\varepsilon_n^2)}} \overbrace{(U + \tilde{\Delta}_n)(U + \tilde{\Delta}_n)^\dagger(U + \tilde{\Delta}_n)}^{=U + 2\tilde{\Delta}_n + U \tilde{\Delta}_n^\dagger U + \mathcal{O}(\tilde{\Delta}_n^2)}{} - U\\
        &= \tilde{\Delta}_n \left(\frac{1}{2} - \frac{3\varepsilon_n}{4}\right) - U\tilde{\Delta}_n^\dagger U \left(\frac{1}{2} + \frac{3\varepsilon_n}{4}\right) + \mathcal{O}(\tilde{\Delta}_n^2) + \mathcal{O}(\varepsilon_n^2).
    \end{align*}
    Now we can finally bound the error to
    \begin{align*}
        \delta_n = ||\Delta_n(t_\infty)|| &\leq \left|\left|\tilde{\Delta}_n\left(\frac{1}{2} - \frac{3\varepsilon_n}{4}\right)\right|\right| + \left|\left|U\tilde{\Delta}_n^\dagger U \left(\frac{1}{2} + \frac{3\varepsilon_n}{4}\right)\right|\right| + \mathcal{O}(\varepsilon_n^2)\\
        &\leq \frac{\varepsilon_n}{2} + \frac{\varepsilon_n}{2} + \mathcal{O}(\varepsilon_n^2) = \varepsilon_n + \mathcal{O}(\varepsilon_n^2),
    \end{align*}
    which straightforwardly extends to $\delta_{\vec{L}}$ with $\varepsilon_{\vec{L}}$ for partial orders.
\end{proof}

\begin{lemma} \label{lemma:error_repetitions}
    The error of $r$ time steps $||U^r - \tilde{\mathcal{A}}_{\vec{L}}^r||$ is bounded by $r$ times the error of a single time step $\delta_{\vec{L}} = ||\Delta_{\vec{L}}|| = ||U - \tilde{\mathcal{A}}_{\vec{L}}||$, up to order $\delta_{\vec{L}}$.
\end{lemma}
\begin{proof}
    We use the definition of $\Delta = U - \tilde{\mathcal{A}}$ and substitute for $\tilde{\mathcal{A}}$.
    \begin{align*}
        ||U^r - \tilde{\mathcal{A}}_{\vec{L}}^r|| &= ||U^r - (U - \Delta)^r||
        = ||U^r - U^r + \sum_{k=1}^{r} U^{k-1}\, \Delta\, U^{r-k} + \mathcal{O}(\Delta^2)|| \\
        &\leq \sum_{k=1}^{r} ||U^{k-1}\, \Delta\, U^{r-k}|| + ||\mathcal{O}(\Delta^2)||
        \leq \sum_{k=1}^{r} ||\Delta|| + \mathcal{O}(\delta^2) = r\delta + \mathcal{O}(\delta^2)
    \end{align*}
\end{proof}

\begin{lemma} \label{lemma:error_bound}
    The logarithmic inverse error bound of a single time step for full orders $\log(\varepsilon_n^{-1})$ scales like
    \begin{equation*}
        \log\frac{1}{\varepsilon_n} = \mathcal{O}(n \log n)
    \end{equation*}
\end{lemma}
\begin{proof}
    We can bound the residual of the Taylor series by
    \begin{align*}
        \varepsilon_n &= \sum_{k=n+1}^\infty \frac{(\overbrace{t_\infty \Lambda}^{\log 2})^k}{k!}\\
        &= \frac{\log^n 2}{n!} \sum_{k=1}^\infty \frac{n! \log^k 2}{(k+n)!}\\
        &\leq \frac{\log^n 2}{n!} e^{\log 2} = \frac{2 \log^n 2}{n!}\\
    \shortintertext{and use Stirling's approximation $n! \leq e \, n^{n+1/2}\, e^{-n}$ to find}
        \varepsilon_n &\leq \frac{2e^{-n} \log^n 2}{e \, n^{n+1/2}}\\
        \log\varepsilon_n &\leq \log 2 - n - n\log\log 2 - 1 - \left(n + \frac{1}{2}\right) \log n.
    \end{align*}
    Therefore
    \begin{equation*}
        \log\frac{1}{\varepsilon_n} = \mathcal{O}(n\log n).
    \end{equation*}
\end{proof}

\begin{corollary} \label{coroll:epsilon_complexity}
    Because the total complexity is of the order $C_n = nL$, the complexity when using full orders depending on the error bound $\varepsilon$, which we will call $C_{\varepsilon, \mathrm{full}}$ scales like
    \begin{equation*}
        C_{\varepsilon, \mathrm{full}} = \mathcal{O}\left(\frac{L\log\frac{1}{\varepsilon}}{\log\log\frac{1}{\varepsilon}}\right)
    \end{equation*}
\end{corollary}
\begin{proof}
    We can use the inequality in \cref{lemma:error_bound}
    \begin{equation*}
        n < \log\frac{1}{\varepsilon_n}
    \end{equation*}
    to replace the $n$ in the logarithm of \cref{lemma:error_bound} and find
    \begin{equation*}
        n = \mathcal{O}\left(\frac{\log\frac{1}{\varepsilon_n}}{\log\log\frac{1}{\varepsilon_n}}\right)
    \end{equation*}
    and therefore
    \begin{equation*}
        C_{\varepsilon, \mathrm{full}} = \mathcal{O}\left(\frac{L\log\frac{1}{\varepsilon}}{\log\log\frac{1}{\varepsilon}}\right).
    \end{equation*}
\end{proof}

\begin{lemma} \label{lemma:error_bound_eps_L}
    The bound for the logarithmic inverse error of the modified version $\log\varepsilon_{\vec{L}}^{-1}$ scales like\\$\mathcal{O}\left(\frac{C_{\vec{L}}}{L}\log \frac{C_{\vec{L}}}{L}\right) \leq \log\varepsilon_{\vec{L}}^{-1}  < \mathcal{O}(C_{\vec{L}} \log C_{\vec{L}})$, depending on the Hamiltonian.
\end{lemma}
\begin{proof}
    The left inequality follows immediately from the worst case that $\alpha_0 = \alpha_1 = \ldots = \alpha_L$. In this case the modification is equivalent to the original method and \cref{lemma:error_bound} with $n = C_{\vec{L}}/L$ holds.

    In the other extreme case of $\frac{\alpha_\ell}{\alpha_0} \to 0 \:\: \forall\: \ell \in \{1\ldots L\}$, one term dominates the whole Hamiltonian, and adding $h_0$ in some order of the expansion equates to adding that whole order, effectively reducing $L$ to 1. Therefore \cref{lemma:error_bound} with $L=1$ defines an upper bound for the error scaling.
\end{proof}

\begin{corollary} \label{coroll:modified_complexity}
    The complexity of our modified version depending on the simulation error bound $\varepsilon$, which we call $C_\varepsilon$, is bounded by
    \begin{equation*}
        \mathcal{O}\left(\frac{\log\frac{1}{\varepsilon}}{\log\log\frac{1}{\varepsilon}}\right) < C_\varepsilon \leq \mathcal{O}\left(\frac{L\log\frac{1}{\varepsilon}}{\log\log\frac{1}{\varepsilon}}\right),
    \end{equation*}
\end{corollary}
\begin{proof}
    The same reasoning as in \cref{coroll:epsilon_complexity} applies to the bounds established in \cref{lemma:error_bound_eps_L}.
\end{proof}

\begin{corollary} \label{coroll:total_error_complexity}
    The complexity of simulating for a time $\tau$ with a total error bound $\epsilon$ is bounded by
    \begin{equation*}
        \mathcal{O}\left(\frac{\Lambda \tau \log\frac{\Lambda\tau}{\epsilon}}{\log\log\frac{\Lambda\tau}{\epsilon}}\right) < C_\epsilon \leq \mathcal{O}\left(\frac{L\Lambda\tau\log\frac{\Lambda\tau}{\epsilon}}{\log\log\frac{\Lambda\tau}{\epsilon}}\right).
    \end{equation*}
\end{corollary}
\begin{proof}
    Simulating for a time $\tau = rt_\infty = r\log(2)/\Lambda$ requires $r$ steps. The error of a single step $\varepsilon$ must therefore be $\varepsilon = \epsilon/r$. Substituting this in \cref{coroll:modified_complexity} and multiplying by the number of steps $r = \mathcal{O}(\tau\Lambda)$  proves the claim.
\end{proof}

\end{document}